\documentclass[10pt,onecolumn,twoside]{article}

\usepackage[margin = 1.0in]{geometry}
\usepackage{mathtools}

\usepackage{amsmath}

\usepackage{graphicx,epsfig,color}
\usepackage{amsfonts}
\usepackage{subfigure}
\usepackage{bm}
\usepackage{amssymb}
\usepackage{amsthm}
\usepackage{epstopdf}
\usepackage{mathtools}  
\usepackage{tabulary}
\usepackage{booktabs}
\usepackage{cite}

\DeclareSymbolFont{bbold}{U}{bbold}{m}{n}
\DeclareSymbolFontAlphabet{\mathbbold}{bbold}
\newcommand{\vect}[1]{\mathbbold{#1}}
\newcommand{\vectorones}[1][]{\vect{1}_{#1}}
\newcommand{\vectorzeros}[1][]{\vect{0}_{#1}}

\def\diag{\operatorname{diag}}

\def\real{\mathbb{R}}

\def\naturals{\mathbb{N}}

\providecommand{\keywords}[1]{\textbf{\textit{Index terms:}} #1}

\newcommand{\until}[1]{\{1,\dots, #1\}}

\newcommand{\subscr}[2]{#1_{\textup{#2}}}

\newcommand{\setdef}[2]{\{#1 \; | \; #2\}}

\newcommand{\map}[3]{#1: #2 \rightarrow #3}

\renewcommand{\baselinestretch}{1.2}

\renewcommand{\theenumi}{\roman{enumi}}

\newcommand{\lmax}{\subscr{\lambda}{max}}
\newcommand{\vmax}{v_{\max}}

\newcommand{\umax}{u_{\max}}
\newcommand{\ex}{\mathrm{e}} 

\newcommand{\hA}{\hat{A}}

\newcommand{\nohA}{ \tfrac{\beta}{\gamma} A}
 
\newcommand{\Ooop}{F_+}
\newcommand{\ooop}{f_+}
\newcommand{\ooom}{f_-}
\newcommand{\Ooom}{F_-}

\DeclareSymbolFont{bbold}{U}{bbold}{m}{n}
\DeclareSymbolFontAlphabet{\mathbbold}{bbold}

\newtheorem{theorem}{Theorem}
\newtheorem{lemma}[theorem]{Lemma}

\renewcommand{\theenumi}{(\roman{enumi}}

\title{On the Dynamics of Deterministic Epidemic Propagation over Networks\thanks{This material is based upon work supported by, or
  in part by, the U.~S.~Army Research Laboratory and the U.~S.~Army
  Research Office under grant numbers W911NF-15-1-0577. }}

\author{\qquad Wenjun Mei \qquad Shadi Mohagheghi \qquad Sandro Zampieri \qquad Francesco Bullo \thanks{ Wenjun Mei and
    Francesco Bullo are with the Department of Mechanical Engineering and
    with Center for Control, Dynamical Systems, and Computation, University
    of California, Santa Barbara, Santa Barbara, CA 93106, USA,
    \texttt{meiwenjunbd@gmail.com}, \texttt{bullo@engineering.ucsb.edu} Shadi Mohagheghi is with the Department of Electrical and Computer Engineering, University of California at Santa Barbara, Santa Barbara, CA 93106, USA,
    \texttt{shadi.mohagheghie@gmail.com} Sandro Zampieri is with the Department of Information Engineering, University of Padova, Italy, \texttt{zampi@dei.unipd.it} }}

\begin{document}
\maketitle

\begin{abstract}
  In this work we review a class of deterministic nonlinear models for
  the propagation of infectious diseases over contact networks with
  strongly-connected topologies. We consider network models for
  susceptible-infected (SI), susceptible-infected-susceptible (SIS),
  and susceptible-infected-recovered (SIR) settings.  In each setting,
  we provide a comprehensive nonlinear analysis of equilibria,
  stability properties, convergence, monotonicity, positivity, and
  threshold conditions.
  For the network SI setting, specific contributions include
  establishing its equilibria, stability, and positivity properties.
  For the network SIS setting, we review a well-known deterministic
  model, provide novel results on the computation and characterization
  of the endemic state (when the system is above the epidemic
  threshold), and present alternative proofs for some of its
  properties.
  Finally, for the network SIR setting, we propose novel results for
  transient behavior, threshold conditions, stability properties, and
  asymptotic convergence. These results are analogous to those
  well-known for the scalar case. In addition, we provide a novel
  iterative algorithm to compute the asymptotic state of the network
  SIR system.
\end{abstract}


\section{Introduction}

\subsection{Motivation and problem description}
Propagation phenomena appear in numerous disciplines. Examples include
the spread of infectious diseases in contact networks, the
transmission of information in communication networks, the diffusion
of innovations in competitive economic networks, cascading failures in
power grids, and the spreading of wild-fires in forests. Scalar models
of propagation phenomena have been widely studied, e.g., see the
survey by Hethcote~\cite{HWH:00} on scalar epidemic spreading models.  These
models qualitatively capture some dynamic features, including phase
transitions and asymptotic states. However, shortcomings of these
scalar models are also prominent: for example, scalar models are
typically based on the assumption that individuals in the population
have the same chances of interacting with each other. This assumption
overlooks the internal structure of the network over which the
propagation occurs, as well as the heterogeneity of individuals in the
network. Both these aspects play critical roles in shaping the
sophisticated dynamical behavior of the propagation processes.

In a general formulation, propagation is a stochastic process on a
complex network. Its features are determined by the properties of
local node-to-node exchanges as well as of the global contact
networks. Stochastic propagation processes can be modeled as Markov
chains of exponential dimensions in the size of the network. Due to
these properties of propagation processes, many relevant research
questions arise naturally. Three key questions are: (1) how to deal
with the lack of well-quantified local details of the complex networks; (2)
how to identify the models and parameters, and estimate the states of
such stochastic large-dimension phenomena, and (3) how to analyze the
transient and asymptotic properties of their dynamics.

Three approaches are commonly adopted in the analysis of propagation
models. The first is to directly characterize the stochastic process
adopting tools from branching processes, percolation, and random graph
theory; e.g., see~\cite{MD-LM:10}. The second approach proposes
a degree-based model based on the approximation that
nodes with the same degree exhibit similar behavior; e.g.,
see~\cite[Chapter~17]{MEJN:10}. The first two approaches are both
based on the random graphs for which global characteristics can be
estimated statistically, such as the degree distribution. The third
approach is based on the mean-field approximation of Markov-chain models
and algebraic graph theory; this approach results in a deterministic
network model. Distinct from the first two approaches, the mean-field
models first assume knowledge of the local propagation parameters and
then relate the dynamical behavior of the propagation process to some
global parameters of the network, which can be estimated without
knowing the full local details, e.g., the spectral radius of the
network's adjacency matrix.  One of the main advantages of the
mean-field approach is that, by representing the network as an
adjacency matrix, well-established theorems in matrix analysis and
dynamical systems can be applied to the analysis of some sophisticated
behavior of the propagation processes.

In this paper we review a class of epidemic propagation models which
adopt the third approach. Distinct in the assumptions on the
microscopic features of the disease and the individual behavior, the
epidemic propagation models we focus on are classified into three
types: the \emph{Susceptible-Infected} (SI) model, the
\emph{Susceptible-Infected-Susceptible} (SIS) model and the
\emph{Susceptible-Infected-Recovered} (SIR) model; basic
representations of these models are illustrated in
Figure~\ref{fig:compartmental}. In this work we review models of these
three types over networks and characterize their dynamical
properties. In short, we study the network SI, network SIS, and
network SIR models.
\begin{figure}[htb]
  \centering
  \includegraphics[width=.5\linewidth]{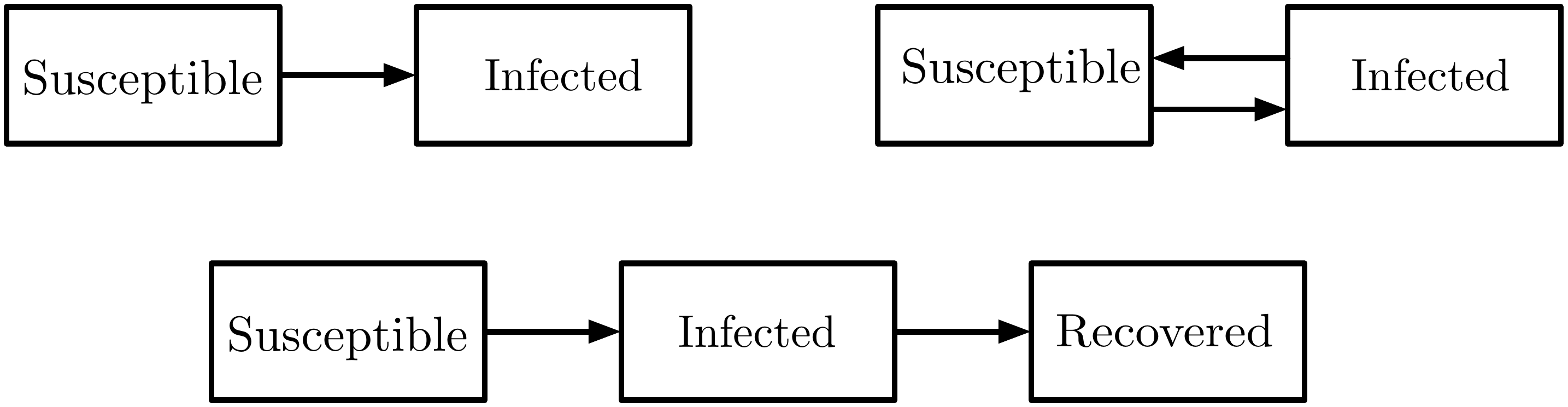}
  \caption{Three basic models of infectious diseases: SI, SIS and SIR.}
  \label{fig:compartmental}
\end{figure}

\subsection{Literature review}

The dynamics of several classic scalar epidemic models, i.e., the
population models without network structure, are surveyed in detail
by Hethcote~\cite{HWH:00}. Among the different metrics discussed, identifying
the \emph{effective reproduction number} $R$ is of particular interest
to researchers; $R$ is the expected number of individuals that a
randomly infected individual can infect during its infection period.
In these scalar models, whether an epidemic outbreak occurs or the
disease dies down depends upon whether $R>1$ or $R<1$, i.e., upon
whether the system is above or below the so-called \emph{epidemic
  threshold}. Here by epidemic outbreak we mean an exponential growth
of the fraction of the infected population for small time. The
\emph{basic reproduction number} $R_0$ is the effective reproduction
number in a fully-healthy susceptible population. In what follows we
focus our review on deterministic network models.

The earliest work on the (continuous-time heterogeneous) SIS model on
networks is~\cite{AL-JAY:76}. This work proposes an $n$-dimensional
model on a contact network and analyzes the system's asymptotic
behavior. This article proposes a rigorous analysis of the threshold
for the epidemic outbreak, which depends on both the disease
parameters and the spectral radius of the contact network.  For the
case when the basic reproduction number is above the epidemic
threshold, this paper establishes the existence and uniqueness of a
nonzero steady-state infection probability, called the endemic
state. In what follows we refer to the model proposed by Lajmanovich et al.
~\cite{AL-JAY:76} as the \emph{network SIS model}; it is also known as the multi-group or
multi-population SIS model. Numerous extensions and variations on
these basic results have appeared over the years.

Allen~\cite{LJSA:94} proposes and analyzes a discrete-time network SIS
model. This work appears to be the first to revisit and formally
reproduce, for the discrete-time case, the earlier results
by Lajmanovich et al.~\cite{AL-JAY:76}; see also the later work by Wang et al.~\cite{YW-DC-CW-CF:03}. This
work confirms the existence of an epidemic threshold, as a function of
the spectral radius of the contact network. Further recent results on
the discrete-time model are obtained by Ahn et al.~\cite{HJA-BH:13} and
by Azizan~Ruhi et al.~\cite{NAR-BH:15}.

Van~Mieghem et al.\cite{PvM-JO-RK:09} argue that the (continuous-time) network SIS model
is in fact the mean-field approximation of the original Markov-chain
SIS model of exponential dimension; this claim is rigorously proven
by Sahneh et al.~\cite{FDS-CS-PvM:13}. Van~Mieghem et al.~\cite{PvM-JO-RK:09} refer to this model as
the intertwined SIS model and write the endemic state as a continued
fraction.

The works by Fall et al.~\cite{AF-AI-GS-JJT:07} and Khanafer et al.~\cite{AK-TB-BG:16} discuss the
continuous-time network SIS model in a more modern
language. Fall et al.~\cite{AF-AI-GS-JJT:07} refer to this model as the $n$-group
SIS model and apply Lyapunov techniques and Metzler matrix theory to
establish existence, uniqueness, and stability of the equilibrium
points below and above the epidemic threshold. Khanafer et al.~\cite{AK-TB-BG:16} use
positive system theory in their analysis and extend the existence,
uniqueness, and stability results to the setting of weakly connected
digraphs.

An early work by Hethcote~\cite{HWH:78} proposes a general multi-group SIR
model with birth, death, immunization, and de-immunization. The
epidemic threshold and the equilibria below/above the threshold are
characterized. For the simplified model without birth/death and
de-immunization, Hethcote~\cite{HWH:78} proves that the system converges
asymptotically to an all-healthy state. Guo et al.~\cite{HG-ML-ZS:08} consider a
generalized network SIR model with vital dynamics, that is, with birth
and death. They characterize the basic reproduction number and,
through a careful Lyapunov analysis, show the existence and global
asymptotic stability of an endemic state above the threshold.
Youssef et al.~\cite{MY-CS:11} study a special case of the network SIR model under
the name of {individual-based SIR model} over undirected
networks. Through a simulation-based analysis, the epidemic threshold
is given as a function of the spectral radius of the network. To the
best of our knowledge, no works have comprehensively characterized the
properties of the network SI model.

We conclude by mentioning other surveys and textbook treatments.
In~\cite{MM-ME:10}, the stability of equilibria for the SEIR model is
reviewed through Lyapunov and graph theory. The additional state $E$
represents the exposed population, i.e., the individuals who are
infected but not infectious. The book chapters~\cite[Chapter~17]{MEJN:10},
\cite[Chapter~21]{DE-JK:10}, and \cite[Chapter~9]{AB-MB-AV:08} review
various heterogeneous epidemic models. The recent survey
by Nowzari et al.~\cite{CN-VMP-GJP:16} presents various epidemic models and addresses
many solved and open problems in the control of epidemic spreading.

\subsection{Statement of Contribution}

The contributions of this work are as follows: in each section, we
start by reviewing the scalar SI, SIS, and SIR models; these are the
models in which variables represent an entire “well-mixed” population
or nodes of an all-to-all unweighted graph. We then focus our
discussion on multi-group network models and provide a tutorial
comprehensive treatment with comprehensive statements and proofs for
the network SI, SIS and SIR models.

We first introduce and analyze the novel network SI model.  We analyze
its asymptotic convergence, positivity of infection probabilities,
initial growth rate, and the stability of equilibria. We show that in
the network SI model, the system does not display a threshold and all
the trajectories converge to the full contagion state.

Next we focus on the network SIS model.  We review some results
in~\cite{AL-JAY:76, AF-AI-GS-JJT:07, AK-TB-BG:16} regarding the
dynamical behavior of the system below and above the threshold, and
present alternative proofs for them. For systems above the epidemic
threshold, we present a novel provably-correct iterative algorithm for
computing the fraction of infected individuals converging to the
endemic state. We present novel Taylor expansions for the endemic
state near the epidemic threshold and in the limit of high infection
rates. Finally, we show that the spread of infection takes place
instantaneously upon infecting at least one node in the network.

Finally, for the network SIR model, we present novel transient
behavior and system properties. We propose new threshold conditions
above which the epidemic grows initially, and below which it
exponentially dies down. We show that, along all system trajectories,
the infected population asymptotically vanishes and the epidemic
asymptotically dies down. The initial rate of growth above the
threshold is given in terms of network characteristics, initial
conditions, and infection parameters. We show that our proposed
weighted average of the infected population, obtained by the entries
of dominant eigenvector of an irreducible quasi-positive matrix,
captures information regarding the distribution of infection in the
system. We also establish positivity of the infection probabilities
and certain monotonicity properties. Moreover, we provide a novel 
iterative algorithm to compute the asymptotic state of the network 
SIR model, with any arbitrary initial condition. For the iterative 
algorithm, the existence and uniqueness of the fixed point, and the 
convergence of the iteration are rigorously proved. Our results are 
analogous to the scalar SIR model properties and are valid for any 
arbitrary network topologies. In comparison with \cite{MY-CS:11}, 
our treatment builds on their numerical results but our result is 
more general in that it does not depend upon specific initial conditions 
and graph topologies, and establishes numerous properties, 
including the novel characterization of epidemic threshold. 

Finally, we remark that our deterministic network models are derived
from the Markov-chain models through a mean-field
approximation. We do not discuss here the Markov-chain model and the
approximation process and refer instead to~\cite{FDS-CS-PvM:13}
and~\cite[Chapter 17]{FB:16}.

\subsection{Organization}
Section~\ref{sec:model-setup} introduces our model set-up and some
preliminary notations. The SI, SIS and SIR models are presented,
respectively, in Sections~\ref{sec:SI-model},~\ref{sec:SIS-model},
and~\ref{sec:SIR-model}. Section 6 is the conclusion.

\section{Model Set-Up and Notations}
\label{sec:model-setup}
For the scalar models, we use the notation $x(t)$ ($s(t)$ and $r(t)$
resp.) for the fraction of infected (susceptible and recovered resp.)
individuals in the population at time $t$. The rest of this section is
about the notations and basic model set-up for the network epidemic
model.

\emph{a) Contact Network:} The epidemics are assumed to propagate over
a weighted digraph $G=(V,E)$, where $V=\{1,\dots,n\}$ and $E$ is the
set of directed links. Nodes of $G$ can be interpreted as either
single individuals in the contact network or as homogeneous
populations of individuals at each location/node in the contact
network. $A=(a_{ij})_{n\times n}$ denotes the adjacency matrix
associated with $G$. For any $i$,~$j\in V$, $a_{ij}$ characterizes the
contact strength from node $j$ to node $i$. For $(i,j)\in E$,
$a_{ij}>0$ and for $(i,j)\notin E$, $a_{ij}=0$. In this paper, $G$ is
assumed to be strongly connected.

\emph{b) Node States and Probabilities:} For different epidemic
propagation models, the set of possible node states are distinct. For
network SI or SIS models, each node can be in either the
``susceptible'' or ``infected'' state, while in the network SIR model,
there is an additional possible node state: ``recovered.'' For a graph
in which the nodes are single individuals, let $s_i(t)$ ($x_i(t)$ and
$r_i(t)$ resp.) be the probability that individual $i$ is in the
susceptible (infected and recovered resp.) state at time
$t$. Alternatively, if the nodes are considered to be the populations,
then $s_i(t)$ ($x_i(t)$ and $r_i(t)$ resp.) is interpreted as the
fraction of susceptible (infected and recovered resp.) individuals in
population $i$. In this paper, without loss of generality, we adopt
the interpretation of nodes as single individuals.

\emph{c) Frequently Used Notations:} 
The symbol $\real$ denotes the set of real numbers, while $\real_{\ge 0}$
denotes the set of non-negative real numbers. The symbol $\phi$ denotes
the empty set. For any two vectors $x,y\in
\real^{n}$, we write
\begin{align*}
  x \ll y, \qquad & \text{if } x_i<y_i \text{ for all $i\in \until{n}$},\\
  x\le y, \qquad & \text{if } x_i\le y_i \text{ for all $i\in \until{n}$}, \text{ and} \\
  x<y, \qquad & \text{if } x\le y \text{ and } x\neq y.
\end{align*}
We adopt the shorthand notations $\vectorones[n] = [1,\dots,1]^{\top}$
and $\vectorzeros[n] = [0,\dots,0]^{\top}$. Given $x =
[x_1,\dots,x_n]^{\top} \in \real^n$, let $\diag(x)$ denote the
diagonal matrix whose diagonal entries are $x_1,\dots,x_n$. For an
irreducible nonnegative matrix $A$, let $\lmax(A)$ denote the dominant
eigenvalue of $A$ that is equal to the spectral radius
$\rho(A)$. Moreover, we let $v_{\max}(A)$ ($u_{\max}(A)$ resp.) denote
the corresponding entry-wise strictly positive left (right resp.)
eigenvector associated with $\lmax(A)$, normalized to satisfy
$\vect{1}_n^\top \vmax(A)=1$ (resp. $\vect{1}_n^\top \umax(A)=1$). The
Perron-Frobenius Theorem for irreducible matrices guarantees that
$\lmax(A)$, $v_{\max}(A)$ and $u_{\max}$ are well defined and
unique. Where not ambiguous, we will drop the $(A)$ argument and, for
example, write
\begin{equation*}
  v_{\max}^{\top}A =\lmax v_{\max}^{\top}\quad \text{and}\quad Au_{\max} = \lmax u_{\max},
\end{equation*}
with $v_{\max}\gg\vectorzeros[n]$ and $\vectorones[n]^\top
v_{\max}=1$; $u_{\max}\gg\vectorzeros[n]$ and $\vectorones[n]^\top
u_{\max}=1$.

\section{Susceptible-Infected Model}
\label{sec:SI-model}
In this section, we first review the classic scalar susceptible-infected (SI) model, and then present and characterize the network SI model. 
\subsection{Scalar SI model}
The scalar SI model assumes that the growth rate of the fraction of
the infected individuals is proportional to the fraction of the
susceptible individuals, multiplied by a so-called \emph{infection rate}
$\beta>0$. The model is given by
\begin{equation}\label{def:lumped-SI}
\dot{x}(t) = \beta s(t) x(t) = \beta \big (1-x(t) \big)x(t),
\end{equation}
and its dynamical behavior is given by the lemma below. 



\begin{lemma}[Dynamical behavior of the SI model]
\label{Dynamical behavior of the SI model} 
  Consider the scalar SI model~\eqref{def:lumped-SI} with $\beta>0$.  The solution from initial
  condition $x(0)=x_0\in[0,1]$ is
  \begin{equation}
    \label{eq:solution-SI-closed-form}
    x(t) = \frac{ x_0 \ex^{\beta t} }{1-x_0 + x_0 \ex^{\beta t}}.
  \end{equation}
  All initial conditions $0<x_0<1$ result in the solution $x(t)$ being monotonically increasing and converging to the unique equilibrium $1$ as
  $t\to\infty$.
\end{lemma}

Solutions to equation~\eqref{def:lumped-SI} with different initial
conditions are plotted in Figure~\ref{fig:lumped-SI-evolution}. The SI
model~\eqref{def:lumped-SI} results in an evolution akin to a logistic
curve, and is also called the logistic equation for population growth.

\begin{figure}
  \centering
  \includegraphics[width=.45\linewidth]{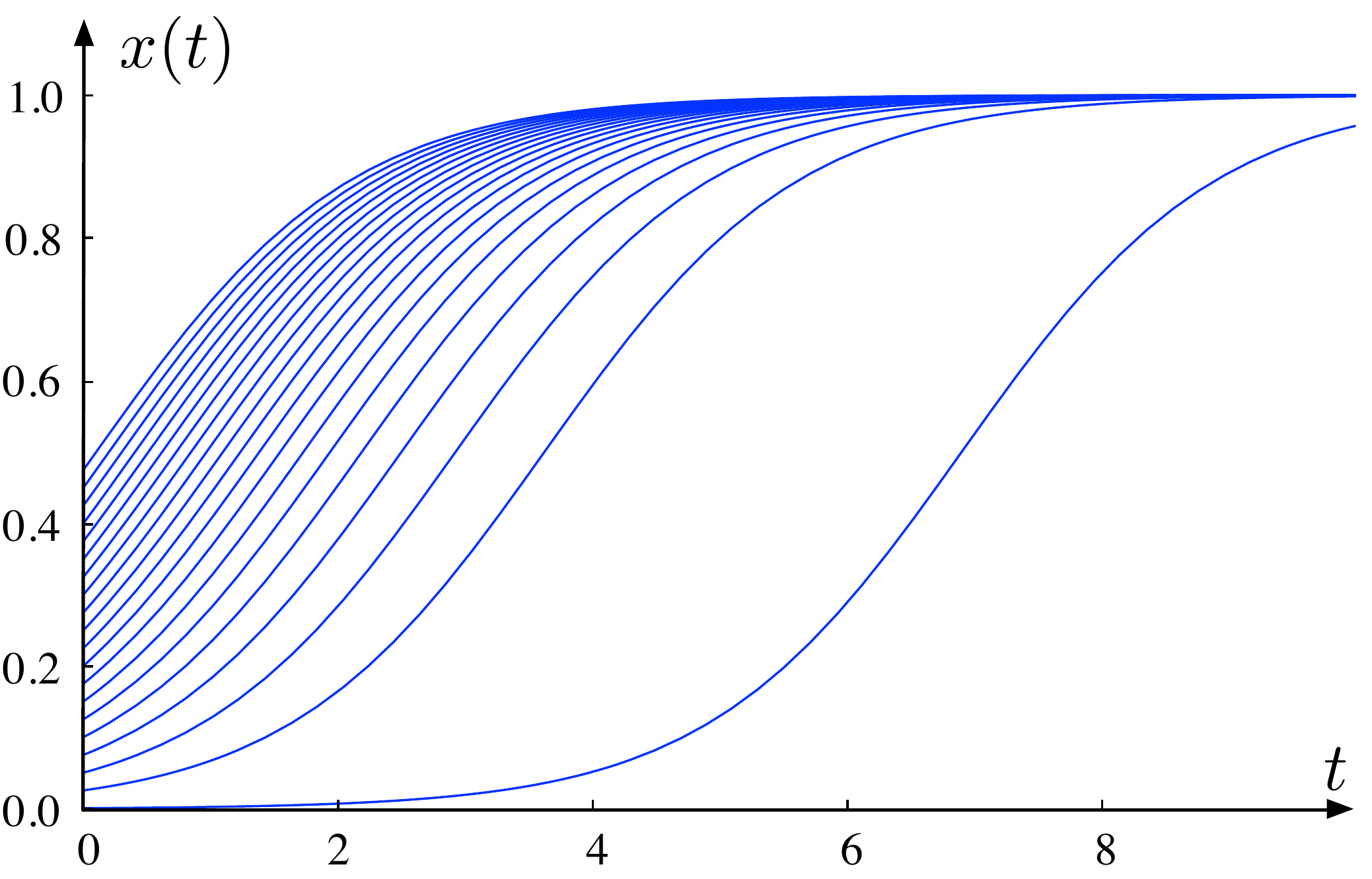}
  \caption{Evolution of the (lumped deterministic) SI model ($\beta=1$)
    from small initial fraction of infected individuals.}
  \label{fig:lumped-SI-evolution}
\end{figure}

\subsection{Network SI model}

The network SI model on a weighted digraph with the adjacency matrix
$A\in \real^{n\times n}_{\ge 0}$ is given by
\begin{equation}\label{def:netSI-entry}
\dot{x}_i(t)=\beta \big( 1-x_i(t) \big)\sum_{j=1}^n a_{ij}x_j(t),
\end{equation} 
or, in equivalent vector form,
\begin{equation}\label{def:network-SI-vector}
  \dot{x}(t)=\beta\Big( I_n-\diag\!\big(x(t)\big) \Big)Ax(t),
\end{equation} 
where $\beta>0$ is the infection rate. Alternatively, in terms of the fractions of susceptibile individuals
$s(t)=\vectorones[n]-x(t)$, the network SI model is
\begin{equation}
  \label{def:network-SI-s}
  \dot{s}(t) = - \beta \diag(s(t)) A (\vectorones[n]-s(t)).
\end{equation}
The following results and their proof are novel.

\begin{theorem}[Dynamical behavior of network SI model]\label{thm:dyn-behav-netSI}
Consider the network SI model~\eqref{def:network-SI-vector} with
$\beta>0$. For strongly connected graph with adjacency matrix $A$, the
following statements hold:
\begin{enumerate}

\item\label{fact:nSI-pos} if $x(0),s(0) \in [0,1]^{n}$, then
  $x(t),s(t) \in [0,1]^{n}$ for all $t > 0$. Moreover, $x(t)$ is
    monotonically non-decreasing (here by monotonically non-decreasing we mean
    $x(t_1)\le x(t_2)$ for all $t_1\le t_2$). Finally, if
  $x(0)>\vectorzeros[n]$, then $x(t)\gg \vectorzeros[n]$ for all $t > 0$;
    
  \item\label{fact:nSI-eq} the model~\eqref{def:network-SI-vector} has two equilibrium points:
    $\vectorzeros[n]$ (no epidemic), and $\vectorones[n]$ (full
    contagion);
    \begin{enumerate}

  \item\label{fact:nSI-lin1} the linearization of
    model~\eqref{def:network-SI-vector} about the equilibrium point
    $\vectorzeros[n]$ is $ \dot{x} = \beta A x$ and it is
    exponentially unstable;

  \item\label{fact:nSI-lin2} let $D=\diag(A\vectorones[n])$ be the
    degree matrix. The linearization of model~\eqref{def:network-SI-s}
    about the equilibrium $\vectorzeros[n]$ is $\dot{s} = -\beta D s$
    and it is exponentially stable;
    \end{enumerate}

  \item\label{fact:nSI-asymp} each trajectory with initial condition
    $x(0)\neq\vectorzeros[n]$ converges asymptotically to
    $\vectorones[n]$, that is, the epidemic spreads monotonically to
    the entire network. 
\end{enumerate}
\end{theorem}

\begin{proof}

\noindent (i) The fact that, if $x(0), s(0)\in[0, 1]^n$, then $x(t),s(t)\in[0, 1]^n$ for all $t>0$ means that $[0, 1]^n$ is an invariant set for the differential equation~\eqref{def:network-SI-vector}. This is the consequence of Nagumo's Theorem (see Theorem 4.7 in \cite{FB-SM:15}), since for any $x$ belonging on the boundary of the set $[0, 1]^n$, the vector 
$ \beta\Big( I_n-\diag\!\big(x\big) \Big)Ax$ is either tangent, or points inside the set $[0, 1]^n$. 

Observe that the invariance of the set $[0,1]^n$ implies that $\dot{x}(t)\ge \vectorzeros[n]$ and so $x(t_1) \le x(t_2)$ for all $t_1 \le t_2$. 

We want to prove now that,if $x(0)>0_n$, then $x(t)\gg 0_n$ for all $t > 0$. If by contradiction there is $i\in\{1,\ldots,n\}$ and $T>0$ such that $x_i(T)=0$, then the monotonicity of $x_i(t)=0$ would imply that $x_i(t)=0$ for all $t\in[0,T]$, which would yield $\dot{x}_i(t)=0$ for all $t\in[0,T]$. By~\eqref{def:netSI-entry} this would imply that $x_j(t)=0$ for all $t\in[0,T]$ for all $j$ such that $a_{ij}>0$. We could iterate this argument and using the irreducibility of $A$ we would get the contradiction that $x(t)=0$ for all $t\in[0,T]$ concluding in this way the proof of~\ref{fact:nSI-pos}.

\noindent (ii) Regarding statement~\ref{fact:nSI-eq}, note that $\vectorzeros[n]$ and $\vectorones[n]$ are clearly equilibrium points. 
Let $\bar x\in[0,1]^n$ be an equilibrium and assume that $\bar x\not=\vectorones[n]$. Then there is $i$ such that $\bar x_i\not=1$. Since $\beta \big( 1-\bar x_i \big)\sum_{j=1}^n a_{ij}\bar x_j=0$, then $\sum_{j=1}^n a_{ij}\bar x_j=0$ which implies that $\bar x_j=0$ for all $j$ such that $a_{ij}>0$. By iterating this argument and using the irreducibility of $A$ we get that $\bar x=0$ concluding only $\vectorzeros[n]$ and $\vectorones[n]$ are equilibrium points. 
Statements~\ref{fact:nSI-lin1} and~\ref{fact:nSI-lin2} are obvious. Exponential stability of the linearization $\dot s =
- \beta Ds$ is obvious, and the Perron-Frobenius Theorem implies the
existence of the unstable positive eigenvalue $\rho(A)>0$ for the
linearization $\dot x = \beta Ax$.

\noindent (iii) Consider the function
$V(x)=\vectorones[n]^\top (\vectorones[n]-x)$; this is a smooth
function defined over the compact and forward invariant set $[0,1]^n$
(see statement~\ref{fact:nSI-pos}). Since $\dot{V}= - \beta \vectorones[n]^\top
\big(I_{n}-\diag(x)\big) A x$, we know that $\dot{V}\leq0$ for all $x$
and $\dot{V}(x)=0$ if and only if
$x\in\{\vectorzeros[n],\vectorones[n]\}$. The LaSalle Invariance
Principle implies that all trajectories with $x(0)$ converge
asymptotically to either $\vectorones[n]$ or
$\vectorzeros[n]$. Additionally, note that $0\leq V(x)\leq n$ for all
$x\in[0,1]^n$, that $V(x)=0$ if and only if $x=\vectorones[n]$ and
that $V(x)=n$ if and only if $x=\vectorzeros[n]$. Therefore, all
trajectories with $x(0)\neq\vectorzeros[n]$ converge asymptotically to
$\vectorones[n]$.
\end{proof}

For the adjacency matrix $A$, there exists a non-singular matrix $T$
such that $A=TJT^{-1}$, where $J$ is the Jordan normal form of
$A$. Since $A$ is non-negative and irreducible, according to
Perron-Frobenius theorem, the first Jordan block
$J_1=(\lmax)_{1\times 1}$ and $\lmax>Re(\lambda_i)$
for any other eigenvalue $\lambda_i$ of $A$. Consider now the onset of
an epidemic in a large population characterized by a small initial
infection $x(0)=x_0$ much smaller than $\vectorones[n]$.  The system
evolution is approximated by $ \dot{x} = \beta A x$.  This
``initial-times'' linear evolution satisfies
\begin{align*}
  x(t) = \ex^{\beta A t}x(0) = T\ex^{\beta J t}T^{-1} x(0)= \ex^{\beta \lmax t} \big( T
  \vect{e}_1\vect{e}_1^{\top}T^{-1}x(0) + o(1) \big),
\end{align*}
where $\vect{e}_1$ is the first standard basis vector in
$\real^n$ and $o(1)$ denotes a time-varying vector that vanishes
as $t\to +\infty$. Let $u_1$ denote the first column of $T$ and let
$v_1^{\top}$ denote the first row of $T^{-1}$. Since $AT=TJ$ and
$T^{-1}A=JT^{-1}$, one can check that $u_1$ ($v_1$ resp.) is the right
(left resp.) eigenvector of $A$ associated with the eigenvalue
$\lmax$. Since $T^{-1}T=I_n$, we have $v_1^{\top}u_1=1$. therefore,
\begin{align}
  \notag x(t) & = \ex^{\beta \lmax t}\big( u_1v_1^{\top}x(0) + o(1) \big)\\
  \label{eq:initial-times-si-evolution}
  & = \ex^{\beta \lmax t} \Big( \frac{v_{\max}^{\top}x(0)}{v_{\max}^{\top}u_{\max}}u_{\max} + o(1) \Big). 
\end{align} 
That is, the epidemic initially experiences exponential growth with rate $\beta \lmax$ and with distribution among the nodes given by the eigenvector $u_{\max}$.

Now suppose that at some time $T$, for all $i$ we have that $x_i(T) = 1 - \epsilon_i$, where each $\epsilon_i$ is much smaller than $1$. Then, for time
$t>T$, the approximated system for $s(t)$ is given by:
\begin{align*}
  \dot s_i(t) = -\beta d_i s_i(t) \quad \implies \quad s_i(t)= \epsilon_i
  \ex^{- \beta d_i (t-T)}.
  \label{eq:final-times-si-evolution}
\end{align*}
From the discussion above, we conclude that the initial infection rate
is proportional to the eigenvector centrality, and the final infection
speed is proportional to the degree centrality.

\section{Susceptible-Infected-Susceptible model}
\label{sec:SIS-model}
In this section we review the Susceptible-Infected-Susceptible (SIS)
epidemic model.  In addition to the existence of an infection process
with rate $\beta>0$, this model assumes that the infected individuals
recover to the susceptible state at so-called \emph{recovery rate}
$\gamma>0$.

\subsection{Scalar SIS model}
In the scalar SIS model, the population is divided into two fractions:
the infected $x(t)$ and the susceptible $s(t)$, with $x(t)+s(t)=1$, obeying the
following dynamics:
\begin{equation}
  \label{def:SIS-model}
  \dot{x}(t) = \beta s(t) x(t) - \gamma x(t) = (\beta-\gamma - \beta x(t)) x(t).
\end{equation}
The dynamical behavior of system~\eqref{def:SIS-model} is given below. 

\begin{lemma}[Dynamical behavior of the SIS model]
  \label{lemma:lumped-SIS}
  For the SIS model~\eqref{def:SIS-model} with $\beta>0$ and $\gamma>0$:
  \begin{enumerate}
  \item\label{fact:scalarSIS:1} the closed-form solution to
    equation~\eqref{def:SIS-model} from initial condition
    $x(0)=x_0\in[0,1]$, for $\beta\neq\gamma$, is
    \begin{equation}
      \label{eq:solution-SIS-closed-form}
      x(t) = \frac{(\beta-\gamma)x_0 }{\beta x_0- \ex^{-(\beta-\gamma)t}( \gamma-\beta(1-x_0))  };
    \end{equation}

  \item\label{fact:scalarSIS:2} if $\beta\leq\gamma$, all trajectories
    converge to the unique equilibrium $x=0$ (i.e., the epidemic
    disappears);

  \item\label{fact:scalarSIS:3} if $\beta>\gamma$, then each
    trajectory from an initial condition $x(0)>0$ converges to the
     exponentially stable equilibrium $x^*=(\beta-\gamma)/\beta$,
    which is called the \emph{endemic state}.
  \end{enumerate}
\end{lemma}

Case~\ref{fact:scalarSIS:3} corresponds to the case in which epidemic
outbreaks take place and a steady-state epidemic contagion persists.
The basic reproduction number in this deterministic scalar SIS model is
given by $R_0=\beta/\gamma$. Simulations regarding to
Lemma~\ref{lemma:lumped-SIS}\ref{fact:scalarSIS:2}~and~\ref{fact:scalarSIS:3}
are shown in Figure~\ref{fig:lumped-SIS-evolution}.

\begin{figure}
  \centering
  \includegraphics[width=.35\linewidth]{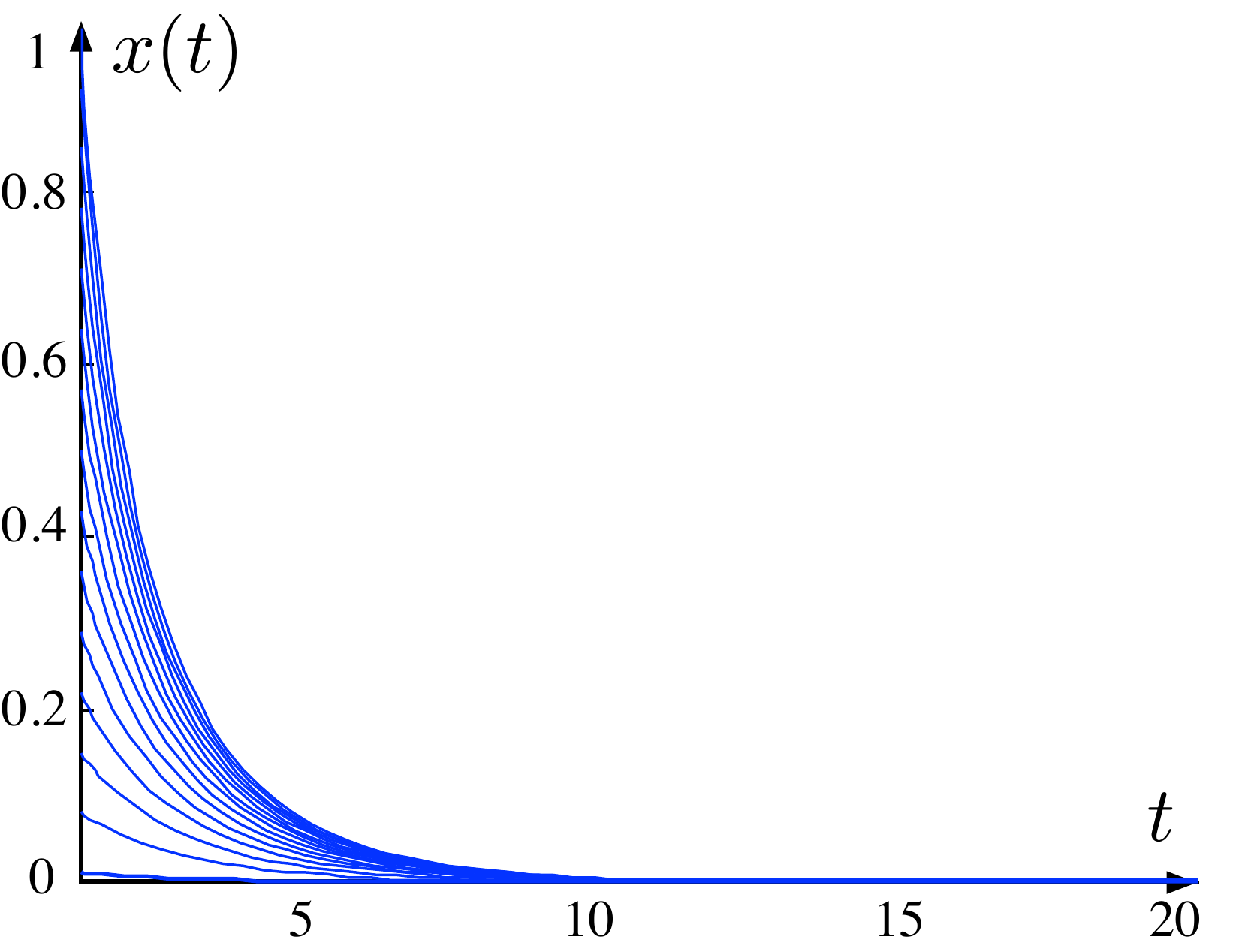}
  \includegraphics[width=.35\linewidth]{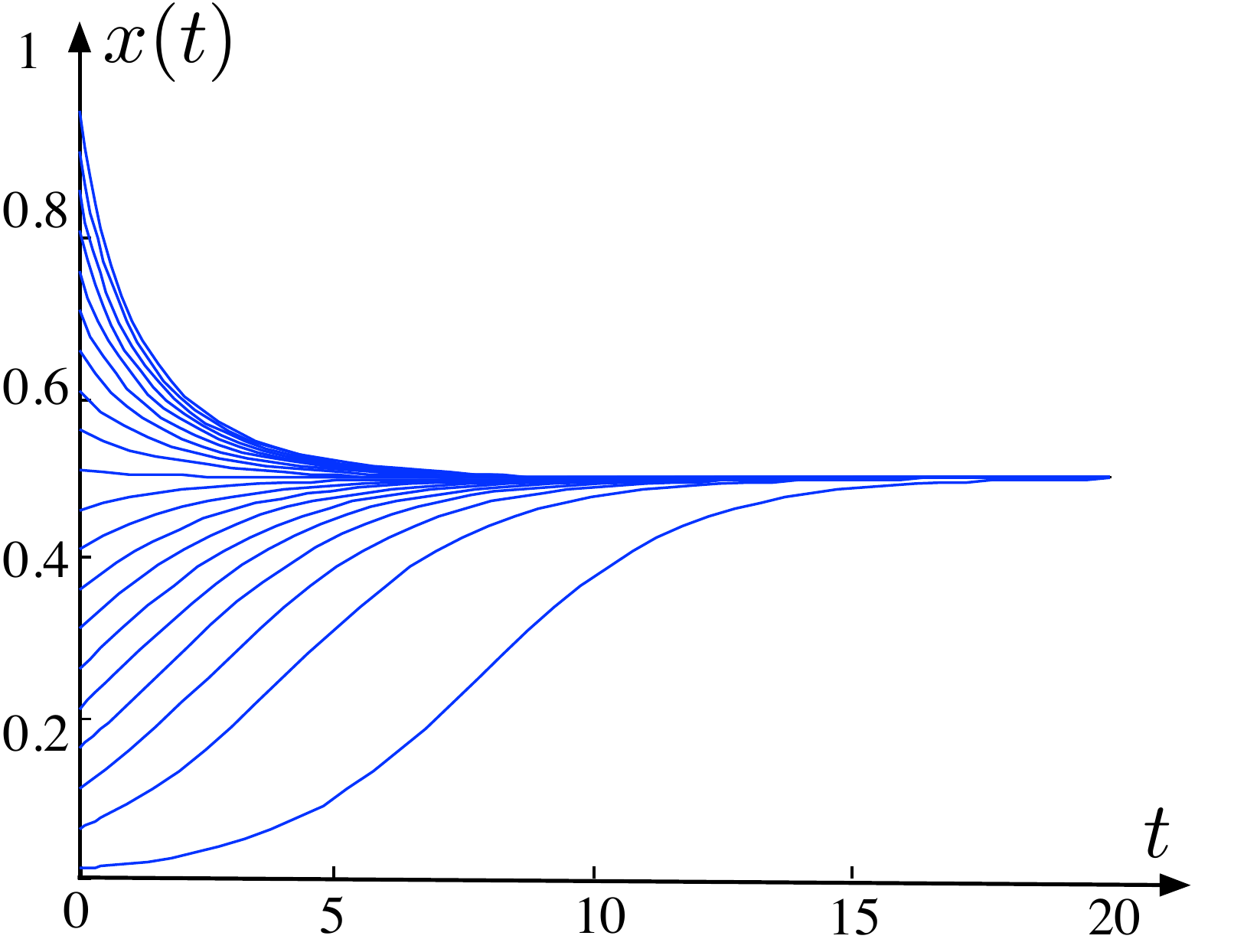}
  \caption{Evolution of the scalar SIS model with varying initial
    fraction of infected individuals. Top figure:
    $\beta=0.5<\gamma=1$. Bottom figure: $\beta=1>\gamma=.5$.}
  \label{fig:lumped-SIS-evolution}
\end{figure}

\subsection{Network SIS Model}
In this section we study the network SIS model which is closely
related to the original ``multi-group SIS model'' proposed
by Lajmanovich~\cite{AL-JAY:76}; see also the intertwined SIS model
in~\cite{PvM-JO-RK:09}.

The network SIS model with infection rate $\beta$ and recovery rate
$\gamma$ is given by:
\begin{equation}
  \label{def:network-SIS}
  \dot{x}_i(t) = \beta (1-x_i(t))\sum_{j=1}^na_{ij} x_j(t) - \gamma x_i(t), 
\end{equation}
or, in equivalent vector form,
\begin{equation}
  \label{def:network-SIS-vector}
  \dot{x}(t) = \beta \big(I_{n}-\diag(x(t))\big) A x(t) - \gamma x(t) .
\end{equation}

In the rest of this section we study the dynamical properties of this
model. We start by defining the monotonically-increasing functions
\begin{equation*}
  \ooop(y)=y/(1+y), \quad\text{and}\quad \ooom(z)=z/(1-z),
\end{equation*}
for $y\in\real_{\geq0}$ and $z\in[0,1[$. Note that $\ooop(\ooom(z))=z$
for all $z\in[0,1)$.  For vector variables $y\in\real_{\geq0}^n$ and
$z\in[0,1)^n$, we write $\Ooop(y)=(\ooop(y_1),\dots,\ooop(y_n))$, and
$\Ooom(z)=(\ooom(z_1),\dots,\ooom(z_n))$.

\paragraph{Behavior of System Below the Threshold}

In this subsection, we characterize the behavior of the network SIS
model in a regime we describe as ``below the threshold.''

\begin{theorem}[Dynamical behavior of the network SIS model: Below the threshold]
  \label{thm:network-SIS-below}
  Consider the network SIS model~\eqref{def:network-SIS}, with
  $\beta>0$ and $\gamma>0$, over a strongly connected digraph with
  adjacency matrix $A$.  Let $\lmax$ and $\vmax$ be the dominant
  eigenvalue of $A$ and the corresponding normalized left eigenvector
  respectively. If $\beta \lmax/\gamma<1$, then
  \begin{enumerate}

   \item\label{fact:belowt-pos} if $x(0),s(0) \in [0,1]^{n}$, then
  $x(t),s(t) \in [0,1]^{n}$ for all $t > 0$. Moreover, if
  $x(0)>\vectorzeros[n]$, then $x(t)\gg \vectorzeros[n]$ for all $t > 0$;

   \item\label{fact:belowt} there exists a unique equilibrium point
     $\vect{0}_n$, the linearization of~\eqref{def:network-SIS} about
     $\vect{0}_n$ is $\dot{x} = (\beta A - \gamma I_n)x$ and it is
     exponentially stable;
    
   \item\label{fact:belowt-globconv} from any $x(0)\neq\vect{0}_n$,
     the weighted average $t\mapsto v_{\max}^\top x(t)$ is
     monotonically and exponentially decreasing, and all the
     trajectories converge to $\vect{0}_n$.
  \end{enumerate}
\end{theorem}

Historically, it is meaningful to attribute this theorem
to~\cite{AL-JAY:76}, even if the language adopted here is more
modern.

\begin{proof}  
\noindent~\ref{fact:belowt-pos}  As in Theorem \ref{thm:dyn-behav-netSI} the first part is the consequence of Nagumo's Theorem. Then define $y(t):=e^{\gamma t}x(t)$. Notice that this variable satisfies the differential equation $\dot{y}(t)=\beta\diag(s(t))Ay(t)$. From the same arguments used in the proof of the point \ref{fact:nSI-pos} of Theorem \ref{thm:dyn-behav-netSI} we argue that $y(t)\gg \vectorzeros[n]$ for all $t>0$. From this it follows that also $x(t)\gg \vectorzeros[n]$ for all $t>0$.

\noindent~\ref{fact:belowt}
Assume that $x^*$ is an equilibrium point. It is easy to se that $x^*\ll \vectorones[n]$. Observe moreover that $x^*$ is an equilibrium point if and only if
$\hA x^*= \Ooom(x^*)$ or, equivalently, if and only if $\Ooop\big(\nohA x^*\big)= x^*$. 
This means that $x^*$ is an equilibrium if and only if it is a fixed point of
$\mathcal{F}$, where $\mathcal{F}(x):=\Ooop\big(\nohA x\big)$.
Let $\hA=\beta A /\gamma$.  For
$x\in[0,1]^n$, note $\Ooop(\hA x)\leq\hA x$
because $\ooop(z)\leq z$.  Moreover, $\vect{0}_n \leq x \leq y$ implies that
$\vect{0}_n \leq \mathcal{F}(x) \leq \hat{A}y$. Therefore, if
$\vect{0}_n \leq x$, then $\mathcal{F}^k (x) \leq \hat{A}^k x$, for
all $k$. Since $\hat{A}$ is Schur stable, then $\lim_{k
  \to\infty}\mathcal{F}^k (x)=0$. This shows that the only fixed point
of $\mathcal{F}$ is zero.

Next, the linearization of equation~\eqref{def:network-SIS-vector} is
verified by dropping the second-order terms. The linearized system is
exponentially stable at $\vect{0}_n$ for $\beta\lmax-\gamma<0$ because
$\lmax$ is larger, in real part, than any other eigenvalue of $A$ by
the Perron-Frobenius Theorem for irreducible matrices.  

\noindent~\ref{fact:belowt-globconv}
Finally, regarding statement~\ref{fact:belowt-globconv}, define
$y(t)=v_{\max}^\top x(t)$ and note that $\big(I_{n}-\diag(z)\big)
v_{\max} \leq v_{\max}$ for any $z\in[0,1]^n$. Therefore,
\begin{align*}
  \dot{y}(t) 
  &\leq 
  \beta \vmax^\top  A x(t) - \gamma \vmax^\top x(t)= 
  (\beta \lmax -\gamma) y(t)<0.
\end{align*}

By the Gr\"onwall-Bellman Comparison Lemma, $y(t)$ is monotonically
decreasing and satisfies $y(t)\leq
y(0)e^{(\beta\lmax-\gamma)t}$ from all initial conditions
$y(0)$. This concludes our proof of
statement~\ref{fact:belowt-globconv}.
\end{proof}

\paragraph{Behavior of System Above the Threshold}
We present the dynamical behavior of the network SIS model above the
threshold as follows.

\begin{theorem}[Dynamical behavior of the network SIS model: Above the threshold]
  \label{thm:network-SIS-above}
Consider the network SIS model~\eqref{def:network-SIS}, with $\beta>0$
and $\gamma>0$, over a strongly connected digraph with adjacency
matrix $A$.  Let $\lmax$ be the dominant eigenvalue of $A$ and let
$\vmax$ and $\umax$ be the corresponding normalized left and right
eigenvectors respectively. Let $d=A\vectorones[n]$. If $\beta
\lmax/\gamma>1$, then

\begin{enumerate}

\item\label{fact:belowt-pos-2} if $x(0),s(0) \in [0,1]^{n}$, then
  $x(t),s(t) \in [0,1]^{n}$ for all $t > 0$. Moreover, if
  $x(0)>\vectorzeros[n]$, then $x(t)\gg \vectorzeros[n]$ for all $t > 0$;
  
\item\label{fact:0unstable} $\vectorzeros[n]$ is an equilibrium point,
  the linearization of system~\eqref{def:network-SIS-vector} at
  $\vectorzeros[n]$ is unstable due to the unstable eigenvalue
  $\beta\lmax-\gamma$ (i.e., there will be an epidemic outbreak);

  \item\label{fact:xstar} besides the equilibrium $\vectorzeros[n]$,
    there exists a unique equilibrium point $x^*$, called the
    \emph{endemic state}, such that
    \begin{enumerate}        
    \item\label{fact:xstar-positive} $x^*\gg\vectorzeros[n]$,

    \item\label{fact:expansion} $x^*=\delta a \umax+O(\delta^2)$ as 
    $\delta\to 0^+$, where $\delta:=\beta\lmax/\gamma-1$ and
      $$a=\frac{\vmax^T \umax}{\vmax^T\diag(\umax)\umax},$$

    \item\label{fact:expansion-near1} $x^*=\vectorones[n] - (\gamma/\beta)
      \diag (d)^{-1} \vectorones[n]+O(\gamma^2/\beta^2)$, at fixed $A$, as
      $\gamma / \beta \to 0^+$,

    \item \label{fact:sequence} define a sequence $\{y(k)\}_{k\in
      \naturals}\subset\real^n$ by
      \begin{equation}\label{eq:net-SIS-eq-iteration}
        y(k+1) := \Ooop\Bigg(\frac{\beta}{\gamma} A y(k) \Bigg).
      \end{equation}
      If $y(0)\ge 0$ is a scalar multiple of $\umax$ and satisfies either
      $0< \max_i y_i(0)\le 1-\gamma/(\beta \lmax)$ or $\min_i
      y_i(0) \ge 1 - \gamma/(\beta \lmax)$, then
      \begin{equation*}
        \lim_{k \to \infty}y(k)=x^*.
      \end{equation*}
      Moreover, if $\max_i y_i(0)\le 1-\gamma/(\beta \lmax)$,
      then $y(k)$ is monotonically non-decreasing; if $\min_i y_i(0)
      \ge 1 - \gamma/(\beta \lmax)$, then $y(k)$ is
      monotonically non-increasing.
   \end{enumerate}

  \item\label{fact:stability+domain} the endemic state $x^*$ is
    locally exponentially stable and its domain of attraction is
    $[0,1]^n\setminus\vectorzeros[n]$.
    
  \end{enumerate}
\end{theorem}

Note: statement~\ref{fact:0unstable} means that, near the onset of an
epidemic outbreak, the exponential growth rate is $\beta\lmax-\gamma$
and the outbreak tends to align with the dominant eigenvector $\umax$;
for more details see the discussion leading up to the approximate
evolution~\eqref{eq:initial-times-si-evolution}.  The basic
reproduction number for this deterministic network SIS model is given
by $R_0=\beta \lmax/\gamma$.

Historically, the existence of a unique endemic state and its global
attractivity properties are due to~\cite{AL-JAY:76}.  To the best of
our knowledge, the Taylor expansions in parts~\ref{fact:expansion}
and~\ref{fact:expansion-near1} and the algorithm in
part~\ref{fact:sequence} are novel.
The proofs of statements~\ref{fact:xstar} based on the properties of
the map $\Ooop$ are novel.

\begin{proof}[Proof of selected statements in Theorem~\ref{thm:network-SIS-above}]

\noindent~\ref{fact:belowt-pos-2} This point can be proved as done in point \ref{fact:nSI-pos} of Theorem \ref{thm:dyn-behav-netSI}.

\noindent~\ref{fact:0unstable} This follows from the same analysis of the
  linearized system as in the proof of
  Theorem~\ref{thm:network-SIS-below}\ref{fact:belowt}.

\noindent~\ref{fact:xstar} We begin by establishing two
  properties of the map $x\mapsto \Ooop(\hA x)$, for $\hA=\beta A
  /\gamma$. First, we claim that, $y \gg z\geq\vectorzeros[n]$
  implies $\Ooop(\hA y) \gg \Ooop(\hA z)$. Indeed, note that $G$ being
  connected implies that the adjacency matrix $A$ has at least one
  strictly positive entry in each row. Hence, $y-z \gg
  \vectorzeros[n]$ implies $\hA(y-z) \gg \vectorzeros[n]$ and, since
  $\ooop$ is monotonically increasing, $\hA y \gg \hA z$ implies
  $\Ooop(\hA y) \gg \Ooop(\hA z)$.
  
  Second, we observe that, for any $0<\alpha <1$ and $z>0$, we have
  $f_+(\alpha z)\ge z$ if and only if $z\le 1-1/\alpha$. Suppose
  $y(0)$ is a scalar multiple of $\umax$ and $0<\max_i y_i(0)\le
  1-\gamma/(\beta \lmax)$. We have
  \begin{equation*}
    F_+(\hat{A}y(0))_i =
    f_+\Big( \frac{\beta \lmax}{\gamma}y_i(0) \Big) \ge y_i(0).
  \end{equation*}
Therefore, the sequence $\{y(k)\}_{k\in \naturals}$ defined by
equation~\eqref{eq:net-SIS-eq-iteration} satisfies $y(1)\ge y(0)$,
which in turn leads to $y(2)=F_+(\hat{A}y(1))\ge
F_+(\hat{A}y(0))=y(1)$, and by induction, $y(k+1)=F_+(\hat{A}y(k))\ge
y(k)$ for any $k\in \naturals$. Such sequence $\{y(t)\}$ is
monotonically non-decreasing and entry-wise upper bounded by
$\vectorones[n]$. Therefore, as $k$ diverges, $y(k)$ converges to some
$x^*\gg \vectorzeros[n]$ such that $F_+\big(\hat{A}x^*\big)=x^*$. This
proves the existence of an equilibrium $x^*=\lim_{k\to \infty}y(k)\gg
\vectorzeros[n]$ as claimed in statements~\ref{fact:xstar-positive}
and~\ref{fact:sequence}.

Similarly, for any $0<\alpha<1$ and $z>0$, $f_+(\alpha z)\le z$ if and
only if $z\ge 1-1/\alpha$. Following the same line of argument in the
previous paragraph, one can check that the $\{y(k)\}_{k\in
  \naturals}$ defined by equation~\eqref{eq:net-SIS-eq-iteration} is
monotonically non-increasing and converges to some $x^*$ , if $y(0)$
is a scalar multiple of $\umax$ and satisfies $\min_i y_i(0) \ge
1-\gamma/(\beta \lmax)$.

Now we establish the uniqueness of the equilibrium
  $x^*\in[0,1]^n\setminus\{\vectorzeros[n]\}$.  First, we claim that an
  equilibrium point with an entry equal to $0$ must be $\vectorzeros[n]$.
  Indeed, assume $y^*$ is an equilibrium point and assume $y_i^*=0$ for
  some $i\in \{1,\dots,n \}$. The equality $y^*_i=\ooop(\sum_{j=1}^n
  a_{ij}y^*_j)$ implies that also any node $j$ with $a_{ij}>0$ must satisfy
  $y^*_j=0$. Because $G$ is connected, all entries of $y^*$ must be zero.
  Second, by contradiction, we assume there exists another equilibrium
  point $y^*\gg \vectorzeros[n]$ distinct from $x^*$.  
Let $\alpha:=\min_j\{y^*_j/x^*_j\}$ and let $i$ such that $\alpha=y^*_i/x^*_i$. Then $y^* \geq \alpha x^*\gg \vectorzeros[n]$ and $y_i^*=\alpha x_i^*$. Notice that we can assume with no loss of generality that $\alpha<1$ otherwise we exchange $x^*$ and $y^*$. Observe now that
  \begin{align*}
    \big(&\Ooop(\hat{A} y^*)  -  y^*\big)_i  
    = \ooop\big((\hat{A} y^*)_i\big) -  \alpha x^*_i  \\
    &\geq \ooop\big(\alpha(\hat{A} x^*)_i\big) -  \alpha x^*_i  
    \qquad \qquad \qquad \text{($\hA \geq \vectorzeros[n\times n]$)}\\
    & > 
    \alpha \ooop\big((\hat{A} x^*)_i\big) -  \alpha x^*_i  
    \qquad \text{($0<\alpha<1$ and $z>0$)} \\
    &= \alpha \big(\Ooop(\hat{A} x^*) -  x^*\big)_i = 0. 
    \qquad \text{($x^*$ is an equilibrium)}
  \end{align*}
  Therefore, $\big(\Ooop(\hat{A} y^*) - y^*\big)_i > 0$, which contradicts the fact that $y^*$ is an equilibrium.
  
Now we prove~\ref{fact:expansion}. Observe first that, since taking 
$$y(0)= \left(1-\frac{\gamma}{\beta\lmax }\right)
\frac{\umax}{\max_i \{{\umax}_{,i}\}}=
\frac{\delta}{\delta+1}\frac{\umax}{\max_i \{{\umax}_{,i}\}}$$
then $y(k)$ is monotonically non-decreasing and converges to $x^*$, and 
since taking instead
$$y(0)= \left(1-\frac{\gamma}{\beta\lmax }\right)
\frac{\umax}{\min_i \{{\umax}_{,i}\}}=
\frac{\delta}{\delta+1}\frac{\umax}{\min_i \{{\umax}_{,i}\}}$$
then $y(k)$ is monotonically non-increasing and converges to $x^*$, we can argue that
$$\frac{\delta}{\delta+1}\frac{\umax}{\max_i \{{\umax}_{,i}\}}
\le x^*\le
\frac{\delta}{\delta+1}\frac{\umax}{\min_i \{{\umax}_{,i}\}}$$
This implies that $x^*$ is infinitesimal as a function of $\delta$. Consider the expansion $x^*(\delta)=x_1\delta+x_2\delta^2+O(\delta^3)$.
Since the equilibrium $x^*$ satisfies the equation 
$$(\delta+1) \big(I_{n}-\diag(x^*)\big) A x^* - \lmax x^*=0$$
by substituting the expansion and equating to zero the coefficient of the term $\delta$ we obtain the equation
$$Ax_1-\lmax x_1=0$$
which proves that $x_1$ is a multiple of $\umax$, namely $x_1=a\umax$ for some constant $a$. 
By equating to zero the coefficient of the term $\delta^2$ we obtain instead the equation
$$Ax_1+Ax_2-\diag(x_1)Ax_1-\lmax x_2=0$$
Using the fact that $x_1=a \umax$ we argue that
$$a \lmax \umax+Ax_2-a^2\lmax\diag(\umax)\umax-\lmax x_2=0$$
By multiplying on the left by $\vmax^T$ we obtain
$$a \lmax \vmax^T \umax-a^2\lmax\vmax^T\diag(\umax)\umax=0$$
which proves that 
$$a=\frac{\vmax^T \umax}{\vmax^T\diag(\umax)\umax}$$

Point ~\ref{fact:expansion-near1} can be proved in a similar way. Indeed, define $\epsilon:=\gamma/\beta$. Since
$$\left(1-\frac{\epsilon}{\lmax}\right)\frac{\umax}{\max_i \{{\umax}_{,i}\}}
\le x^*\le
\left(1-\frac{\epsilon}{\lmax}\right)\frac{\umax}{\min_i \{{\umax}_{,i}\}}$$
we can argue that the expansion $x^*(\epsilon)=x_0+x_1\epsilon+O(\epsilon^2)$ as $\epsilon$ tends to zero is such that $x_0\gg\vectorzeros[n]$.
Since the equilibrium $x^*$ satisfies the equation 
$$\big(I_{n}-\diag(x^*)\big) A x^* - \epsilon x^*=0$$
by substituting the expansion and equating to zero the coefficient of the term $\epsilon^0$ we obtain the equation
$$Ax_0-\diag(x_0) A x_0=0$$
which proves that $x_0=vectorones[n]$. 
By equating to zero the coefficient of the term $\epsilon^1$ we obtain instead the equation
$$Ax_1-\diag(x_1) A x_0-\diag(x_0) A x_1-x_0=0$$
Using the fact that $x_0=\vectorones[n]$ we argue that
$$\diag(A \vectorones[n]) x_1+\vectorones[n]=0$$
which yieds the thesis.
     
\noindent~\ref{fact:stability+domain} For this point we refer to
   \cite{AL-JAY:76,AF-AI-GS-JJT:07} or~\cite[Theorems~1
     and~2]{AK-TB-BG:16} in the interest of brevity.
\end{proof}

\section{Network Susceptible-Infected-Recovered Model}
\label{sec:SIR-model}
In this section we review the Susceptible-Infected-Susceptible (SIR)
epidemic model.

\subsection{Scalar SIR model}
In this model individuals who recover from infection are assumed not
susceptible to the epidemic any more. In this case, the population is
divided into three distinct groups: $s(t)$, $x(t)$, and $r(t)$,
denoting the fraction of susceptible, infected, and recovered
individuals, respectively, with $s(t)+x(t)+r(t)=1$. We write the
(Susceptible--Infected--Recovered) SIR model as:
\begin{equation} \label{def:SIR-model}
\begin{split}
  \dot{s}(t) &= -\beta s(t) x(t), \\
  \dot{x}(t) &= \beta s(t) x(t) - \gamma x(t), \\
  \dot{r}(t) &= \gamma x(t),
\end{split}
\end{equation}
and present its dynamical behavior in the lemma below.

\begin{lemma}[Dynamical behavior of the SIR model]
  \label{lemma:scalar-SIR-model}
  Consider the SIR model~\eqref{def:SIR-model}.  From each initial
  condition $s(0)+x(0)+r(0)=1$ with $s(0)>0$, $x(0)>0$ and
  $r(0)\geq0$, the resulting trajectory $t\mapsto(s(t),x(t),r(t))$ has
  the following properties:
  \begin{enumerate} 
  
  \item\label{fact:sir-lumped-wellposed-1} $s(t)>0$, $x(t)>0$,
    $r(t)\geq0$, and $s(t)+x(t)+r(t)=1$ for all $t\geq0$;
	
  \item\label{fact:sir-lumped-monotonic-2} $t\mapsto s(t)$ is
    monotonically decreasing and $t\mapsto r(t)$ is monotonically
    increasing;

  \item\label{fact:sir-lumped-uniquesolution-5} $\lim_{t\to\infty} (s(t),x(t),r(t)) =
    (s_\infty,0,r_\infty)$, where $r_\infty$ is the unique solution to
    the equality
    \begin{equation}
      \label{def:SIR-rinfty}
      1-r_\infty = s(0) \ex^{ -\frac{\beta}{\gamma} \, \big(r_\infty -r(0)\big) };
    \end{equation}

  \item\label{fact:sir-lumped-belowthreshold-3} if $\beta
    s(0)/\gamma<1$, then $t\mapsto x(t)$ monotonically and
    exponentially decreases to zero as $t\to\infty$;

  \item\label{fact:sir-lumped-abovethreshold-4} if $\beta
    s(0)/\gamma>1$, then $t\mapsto x(t)$ first monotonically increases
    to a maximum value and then monotonically decreases to $0$ as
    $t\to\infty$; the maximum fraction of infected individuals is
    given by:
    \begin{equation*}
      \label{def:max-infected}
      x_{\max} = x(0) + s(0) - \frac{\gamma}{\beta} \Big(\log(s(0))+1-\log\Big(\frac{\gamma}{\beta}\Big) \Big).
    \end{equation*}
  \end{enumerate}

\end{lemma}

As mentioned before, we describe the behavior in
statement~\ref{fact:sir-lumped-abovethreshold-4} as an epidemic
outbreak, an exponential growth of $t\mapsto x(t)$ for small times.)
The effective reproduction number in the deterministic scalar SIR
model is $R=\beta s(t)/\gamma$.  Note that the basic reproduction
number $R_0=\beta/\gamma$ does not have predict power in this model.

\begin{figure}
  \centering
  \includegraphics[width=.32\linewidth]{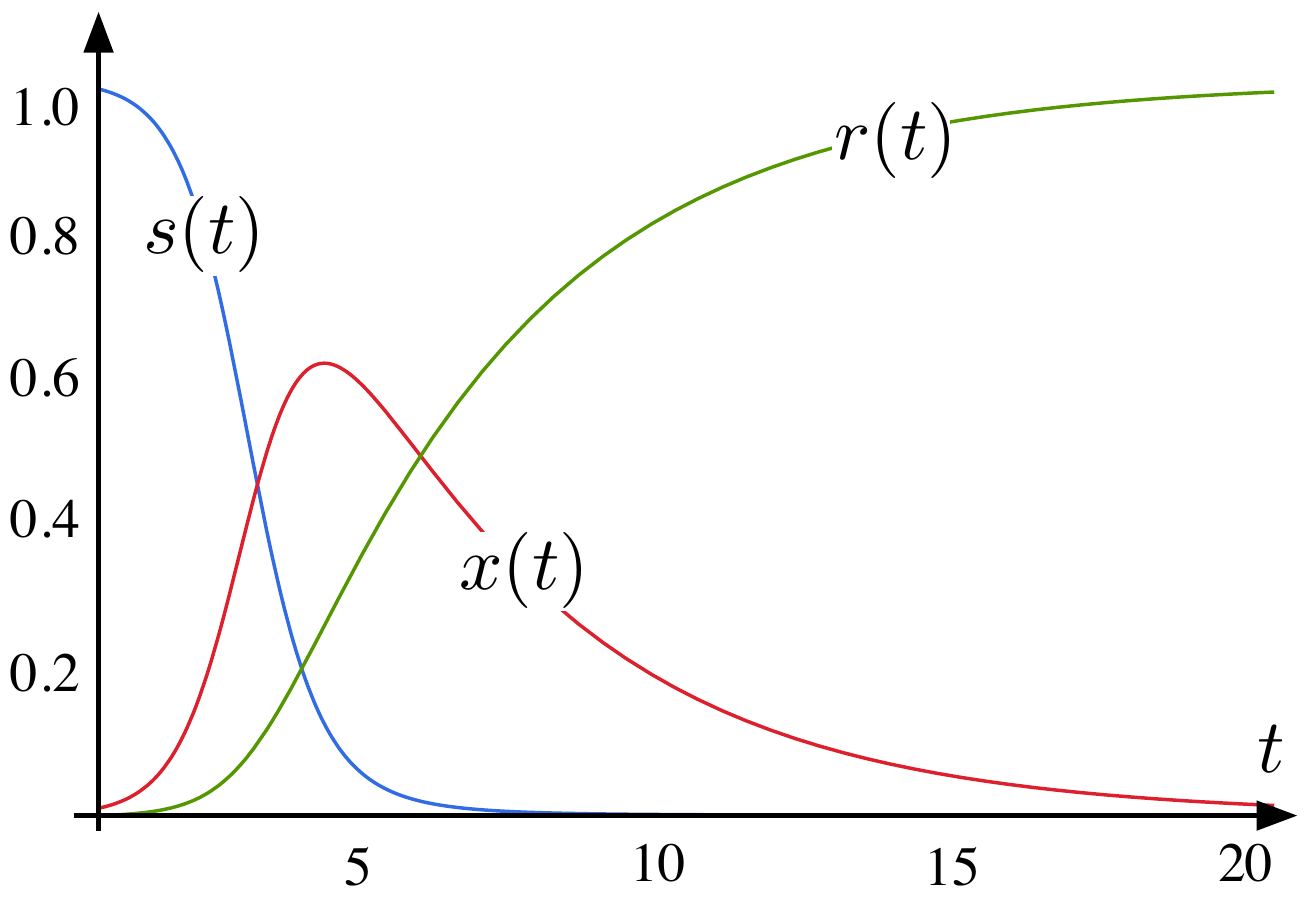} \qquad\qquad
  \hspace{-0.12cm}\includegraphics[width=.35\linewidth]{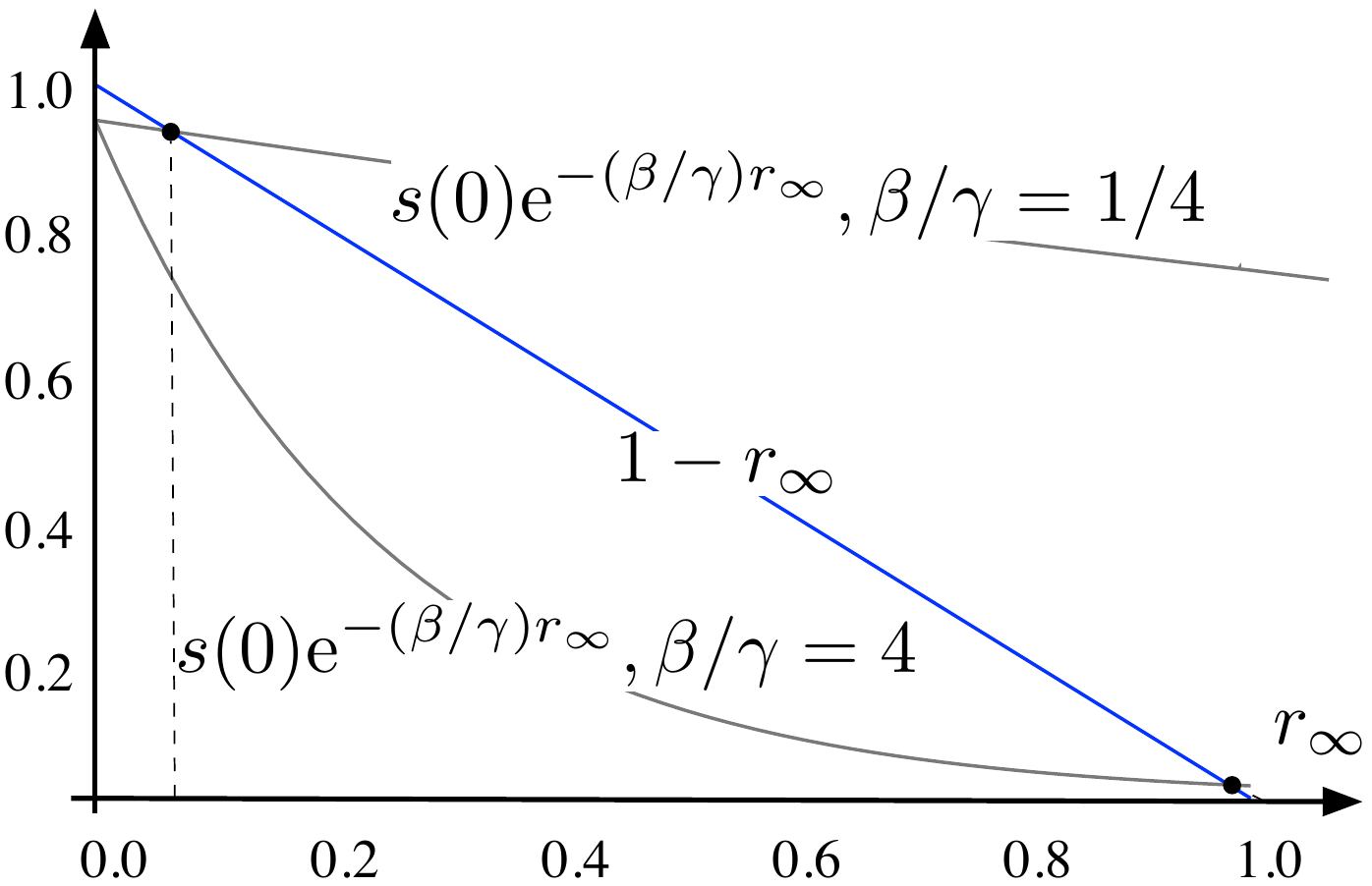}
       \caption{Left figure: evolution of the scalar SIR model from
         small initial fraction of infected individuals (and zero
         recovered); parameters $\beta=2$, $\gamma=1/4$ (case
         \ref{fact:sir-lumped-belowthreshold-3} in
         Lemma~\ref{lemma:scalar-SIR-model}).  Right figure:
         intersection between the two curves in
         equation~\eqref{def:SIR-rinfty} with $s(0)=0.95$, $r(0)=0$
         and $\beta/\gamma\in\{1/4,4\}$. If $\beta/\gamma= 1/4$, then
         $.05<r_\infty<.1$. If $\beta/\gamma=4$, then $.95<r_\infty$.
       }
  \label{fig:lumped-SIR-evolution}
\end{figure}

\subsection{Network SIR model}
The network SIR model on a graph with adjacency matrix $A$ is given by
\begin{align*}
  \dot{s}_i(t) &= -\beta s_i(t)\sum\nolimits_{j=1}^na_{ij} x_j(t),
  \\
  \dot{x}_i(t) &= \beta s_i(t)\sum\nolimits_{j=1}^na_{ij} x_j(t) - \gamma x_i(t),
  \\
  \dot{r}_i(t) &= \gamma x_i(t),
\end{align*}
where $\beta>0$ is the {infection rate} and $\gamma>0$ is the {recovery
  rate}. Note that the third equation is redundant because of the
constraint $s_i(t)+x_i(t)+r_i(t)=1$. Therefore, we regard the
dynamical system in vector form as:
\begin{subequations}\label{def:network-SIR-vector}\begin{align}
    \dot{s}(t) &= -\beta\diag(s(t))Ax(t), \label{def:network-SIR-vector-s} \\
    \dot{x}(t) &= \beta\diag(s(t))Ax(t) - \gamma x(t). \label{def:network-SIR-vector-x}
  \end{align}
\end{subequations}

We state our main novel results of this section below.

\begin{theorem}[Dynamical behavior of the network SIR model]\label{thm:netSIR-asym-behav}
Consider the network SIR model~\eqref{def:network-SIR-vector}, with
$\beta>0$ and $\gamma>0$, over a strongly connected digraph with
adjacency matrix $A$. For $t\geq0$, let $\lmax(t)$ and $\vmax(t)$ be
the dominant eigenvalue of the non-negative matrix $\diag(s(t))A$ and
the corresponding normalized left eigenvector, respectively.  The
following statements hold:
\begin{enumerate}
\item \label{fact:nSIR-0} if $x(0) > \vectorzeros[n]$, and $s(0) \gg \vectorzeros[n]$, then
  \begin{enumerate} 
  \item \label{fact:nSIR-pos} $t \mapsto s(t)$ and $t \mapsto x(t)$ are strictly positive for all $t > 0$,
  \item \label{fact:nSIR-s} $t \mapsto s(t)$ is monotonically decreasing, and
  \item \label{fact:nSIR-lambda} $t \mapsto \lmax(t)$ is monotonically decreasing;
  \end{enumerate}
  
    \item \label{fact:nSIR-equilibrium} the set of equilibrium points
      is the set of pairs $(s^*, \vectorzeros[n])$, for any $s^* \in
      [0, 1]^n$, and the linearization of
      model~\eqref{def:network-SIR-vector} about $(s^*,
      \vectorzeros[n])$ is
    \begin{equation}\label{def:network-SIR-vector-linealized}
    \begin{aligned}
        & \dot{s}(t)=-\beta \diag \big(s^*\big)Ax,\\
        & \dot{x}(t)=\beta \diag \big(s^*\big)Ax-\gamma x;
    \end{aligned}
    \end{equation}
    
   \item \label{fact:nSIR-below}(behavior below the threshold) let
      the time $\tau \geq 0$ satisfy $\beta \lmax(\tau) <
      \gamma$. Then the weighted average $t \mapsto \vmax(\tau)^{\top}
      x(t)$ , for $t \geq \tau$, is monotonically and exponentially
      decreasing to zero;

  \item \label{fact:nSIR-above} (behavior above the threshold) if
    $\beta \lmax(0) > \gamma $ and $x(0)>\vectorzeros[n]$, then,
    \begin{enumerate}
    \item \label{fact:nSIR-inc}(epidemic outbreak) for small time, the
      weighted average $t \mapsto \vmax(0)^{\top} x(t)$ grows
      exponentially fast with rate $\beta\lmax(0)-\gamma$, and
            
    \item \label{fact:nSIR-lambdaT} there exists $\tau>0$ such that
      $\beta \lmax(\tau) < \gamma$;
    \end{enumerate}

    
    \item \label{fact:nSIR-conv} each trajectory converges
      asymptotically to an equilibrium point, that is, $\lim_{t
        \mapsto \infty } x(t)= \vectorzeros[n]$ so that the epidemic
      asymptotically disappears.
\end{enumerate}
\end{theorem}

The effective reproduction number in the deterministic network SIR
model is $R(t)=\beta \lmax(t)/\gamma$.  When $R(0)>1$, we have an
epidemic outbreak, i.e., an exponential growth of infected individual
for short time. In any case, the theorem guarantees that, after at
most finite time, $R(t)<1$ and the infected population decreases
exponentially fast to zero.

\begin{proof}  
Regarding statement~\ref{fact:nSIR-pos}, $s(t)\gg \vectorzeros[n]$ is
due to the fact that $Ax$ is bounded and $s(t)$ is continuously
differentiable to $t$. The statement that $x(t)\gg \vectorzeros[n]$
for all $t>0$ is proved in the same way as
Theorem~\ref{thm:network-SIS-below}~\ref{fact:belowt-pos}. Statement~\ref{fact:nSIR-s}
is the immediate consequence of $\dot{s}_i(t)$ being strictly
negative. From statement~\ref{fact:nSIR-pos} we know that each
$s_{i}(t)$ is positive, and from $A$ being irreducible and $x(0)\neq
\vectorzeros[n]$ we know that $\sum_{j=1}^{n}a_{ij} x_{j}$ is
positive. Therefore, $\dot{s}_i(t) =-\beta s_i(t)\sum_{j=1}^{n}a_{ij}
x_j(t) < 0$ for all $i \in V$ and $t \geq 0$.

For statement~\ref{fact:nSIR-lambda}, we start by recalling the
following property from~\cite[Example~7.10.2]{CDM:01}: for $B$ and
$C$ nonnegative square matrices, if $B\leq C$, then
$\rho(B)\leq\rho(C)$.  Now, pick two time instances $t_1$ and $t_2$ with
$0<t_1<t_2$. Let $\alpha = \max_i s_i(t_2)/s_i(t_1)$ and note
$0<\alpha<1$ because $s(t)$ is strictly positive and monotonically
decreasing. Now note that,
\begin{equation*}
  \diag(s(t_1))A > \alpha  \diag(s(t_1))A \geq \diag(s(t_2))A,
\end{equation*}
so that, using the property above, we know
\begin{equation*}
  \rho(\diag(s(t_1))A) > \alpha  \rho(\diag(s(t_1))A) \geq \rho(\diag(s(t_2))A).
\end{equation*}
This concludes the proof of statement~\ref{fact:nSIR-lambda}.

Regarding statement~\ref{fact:nSIR-equilibrium}, note that a point
$(s^*, x^*)$ is an equilibrium if and only if:
\begin{align*}
  \vectorzeros[n] & =-\beta \diag \big(s^*\big)Ax^*,\quad \text{and}\\
  \vectorzeros[n] & =\beta \diag \big(s^*\big)Ax^*-\gamma x^*.
\end{align*}
Therefore, each point of the form $(s^*, \vectorzeros[n])$ is an
equilibrium. On the other hand, summing the last two equalities we
obtain $\vectorzeros[n]= \gamma x^*$ and thus $x^*$ must be
$\vectorzeros[n]$. As a straightforward result, the linearization of
model~\eqref{def:network-SIR-vector} about any equilibrium point
$(s^*, \vectorzeros[n], \vectorones[n]-s^*)$ is given by
equation~\eqref{def:network-SIR-vector-linealized}.
    
Regarding statement~\ref{fact:nSIR-below}, multiplying
$\vmax(\tau)^{\top}$ from the left on both sides of
equation~\eqref{def:network-SIR-vector-x} we obtain:
\begin{equation*}
  \begin{split}
    \frac{d}{dt}\big(\vmax(\tau)^{\top} x(t) \big) & = \vmax(\tau)^{\top}\Big( \beta \diag\big( s(t) \big)Ax(t) -\gamma x(t) \Big)\\
    & \leq \vmax(\tau)^{\top}\Big( \beta \diag\big( s(\tau) \big)Ax(t) -\gamma x(t) \Big)=(\beta \lmax(\tau) -\gamma) \vmax(\tau)^{\top} x(t) .
  \end{split}
\end{equation*}
Therefore, we obtain
\begin{equation*}
  \vmax(\tau)^{\top} x(t) \le (\vmax(\tau)^\top x(0))
  \ex^{(\beta\lmax(\tau)-\gamma) t}.
\end{equation*}
The right-hand side exponentially decays to zero when $\beta
\lmax(\tau) < \gamma$. Therefore, $\vmax(\tau)^{\top} x(t)$ also
decreases monotonically and exponentially to zero for all $t > \tau$.
    
Regarding statement~\ref{fact:nSIR-inc}, note that based on the
argument in~\ref{fact:nSIR-pos}, we only need to consider the case
when $x(0) \gg \vectorzeros[n]$.  Left-multiplying $\vmax(0)^{\top}$
on both sides of equation~\eqref{def:network-SIR-vector-x}, we obtain:
\begin{equation*}
     \frac{d}{dt}\big( \vmax(0)^{\top}  x(t) \big) \Big \rvert_{t=0} = \vmax(0)^{\top}\Big( \beta \diag\big( s(t) \big)Ax(t) -\gamma x(t) \Big) \Big \rvert_{t=0} =(\beta \lmax(0) -\gamma) \vmax(0)^{\top} x(0).
\end{equation*}
Since $\beta \lmax(0) -\gamma>0$, the initial time derivative
of $ \vmax(0)^{\top} x(t)$ is positive.  Since $t\mapsto
\vmax(0)^{\top} x(t)$ is a continuously differentiable function,
there exists $\tau'>0$ such that $\frac{d}{dt}\big( \vmax(0)^{\top}
x(t)\big)>0$ for any $t\in [0,\tau']$.

Regarding statement~\ref{fact:nSIR-lambdaT}, since $\dot{s}(t)\le
\vectorzeros[n]$ and is lower bounded by $\vectorzeros[n]$, we
conclude that the limit $\lim\limits_{t\to+\infty}s(t)$ exists. Moreover,
since $s(t)$ is monotonically non-increasing, we have $\lim\limits_{t\to
  +\infty}\dot{s}(t)=0$, which implies either $\lim\limits_{t \to
  +\infty}s(t)=\vectorzeros[n]$ or $\lim\limits_{t\to
  +\infty}x(t)=\vectorzeros[n]$. If $s(t)$ converges to
$\vectorzeros[n]$, then $\dot{x}(t)$ converges to $-\gamma
x(t)$. Therefore, there exists $T>0$ such that $\beta
\lmax(T)<\gamma$, which leads to $x(t)\to \vectorzeros[n]$ as
$t\to +\infty$; If $s(t)$ converges to some $s^*>\vectorzeros[n]$,
then $x(t)$ still converges to $\vectorzeros[n]$. Therefore, for any
$\big( s(0),x(0) \big)$, the trajectory $\big( s(t),x(t) \big)$
converges to some equilibria with the form $(s^*,\vectorzeros[n])$,
where $s^*\ge \vectorzeros[n]$. Let
\begin{equation*}
  s(t)=s^* + \delta_s(t),\quad \text{and }x(t)=\vectorzeros[n]+\delta_x(t).
\end{equation*}
We know that $\delta_s(t)\ge 0$ and $\delta_x(t)\ge 0$ for all $t\ge
0$. Moreover, $\delta_s(t)$ is monotonically non-increasing and
converges to $\vectorzeros[n]$, and there exists $\tilde{T}>0$ such
that, for any $t\ge T$, $\delta_x(t)$ is monotonically non-increasing
and converges to $\vectorzeros[n]$.

Let $\lambda^*$ and $v^*$ denote the dominant eigenvalue and the
corresponding normalized left eigenvector of matrix $\diag(s^*)A$,
respectively, that is, ${v^*}^{\top}\diag(s^*)A=\lambda^*
{v^*}^{\top}$. First let us suppose $\beta \lambda^* -\gamma >0$, then
the linearized system of~\eqref{def:SIR-model} around
$(s^*,\vectorzeros[n])$ is written as
\begin{align*}
\dot{\delta}_s & = -\beta \diag(s^*)A\delta_x,\\
\dot{\delta}_x & = \beta\diag(s^*)A\delta_x - \gamma \delta_x.
\end{align*}
Since $\beta \lambda^* -\gamma >0$, the linearized system is
exponentially unstable, which contradicts the fact that
$\big(\delta_s(t),\delta_x(t)\big)\to
(\vectorzeros[n],\vectorzeros[n])$ as $t\to +\infty$. Alternatively,
suppose $\beta \lambda^* -\gamma=0$. By left multiplying
${v^*}^{\top}$ on both sides of the equation for $\dot{x}(t)$
in~\eqref{def:SIR-model}, we obtain
\begin{align*}
{v^*}^{\top}\dot{\delta}_x = (\beta \lambda^*-\gamma)({v^*}^{\top}\delta_x) + \beta {v^*}^{\top}\diag(\delta_s)A\delta_x= \beta {v^*}^{\top}\diag(\delta_s)A\delta_x \ge \vectorzeros[n],
\end{align*}
which contradicts $\delta_x(t)\to \vectorzeros[n]$ as $t\to
+\infty$. Therefore, we conclude that $\beta \lambda^*
-\gamma<0$. Since $\lmax(t)$ is continuous on $t$, we conclude
that there exists $\tau<+\infty$ such that $\beta
\lmax(t)-\gamma<0$.

\end{proof}

In what follows, we present an iterative algorithm that computes the
limit state $\lim_{t\to\infty}\big(s(t),0,r(t)\big)$ of the network
SIR model~\eqref{def:network-SIR-vector} as a function of an arbitrary
initial condition $\big(s(0),x(0),r(0)\big)$. To our best knowledge,
this problem and its solution are novel.

Note that, for the scalar SIR model~\eqref{def:SIR-model}, if we
define
\begin{equation*}
   V\big( s(t), x(t) \big) := s(t) \ex^{ \frac{\beta}{\gamma} \big(1-x(t) -s(t)\big)}.
\end{equation*}
Simple calculations result in $dV\big( s(t),x(t) \big)/dt = 0$, which
implies that the trajectories are on the level sets of $V$ and in the
set $\setdef{(s,x)\in\real^2}{s\geq0, x\geq0, s+x\leq1}$. Here, we
apply a similar approach to the network SIR
system~\eqref{def:network-SIR-vector}. Let
\begin{equation*}
  V_i(s,r):= s_i \ex^{\frac{\beta}{\gamma} \sum_{j=1}^n a_{ij}r_j},\quad \text{for any }i\in \{1,\dots, n\}.
\end{equation*}
One can check that, along any trajectory of
dynamics~\eqref{def:network-SIR-vector}, $dV_i/dt=0$ for any
$i\in\until{n}$. Therefore, the trajectories $(s(t),r(t))$ lie on the
level curves of the functions $V_i(s,r)$ for $i\in\until{n}$.

Let $s(\infty):=\lim_{t\to +\infty}s(t)$, $x(\infty):=\lim_{t\to
  +\infty}x(t)$, and $r(\infty):=\lim_{t\to +\infty}r(t)$. Notice that
$x(\infty)=\vectorzeros[n]$ and so
$r(\infty)=\vectorones[n]-s(\infty)$. Since $dV_i/dt=0$ for any $i\in
\{1,\dots,n\}$, we have
\begin{equation}\label{eq-net-SIR-iteration-asym-s}
  s_i(\infty) = s_i(0)\ex^{ -\frac{\beta}{\gamma}\sum_{j=1}^n a_{ij}\big( 1-r_j(0) \big) } \ex^{ \frac{\beta}{\gamma}\sum_{j=1}^n a_{ij}s_j(\infty) }.
\end{equation}
Given any initial condition $\big( s(0),r(0) \big)$, the right-hand
side of equation~\eqref{eq-net-SIR-iteration-asym-s} defines a map
\begin{equation}\label{eq:SIR-mapH}
H(s):=e^{\frac{\beta}{\gamma}\diag\!\big( A(s-\vectorones[n]+r(0)) \big)}s(0),
\end{equation}
and $s(\infty)$ is a fixed point of $H$, that
is, $s(\infty)=H\big(s(\infty)\big)$.

\begin{theorem}[Existence, uniqueness, and algorithm for the asymptotic point]
  \label{thm:netSIR-algorithm}
  Consider the network SIR model~\eqref{def:network-SIR-vector}, with
  positive rates $\beta$ and $\gamma$ and with initial condition
  $\big(s(0),x(0),r(0)\big)$ satisfying $s(0)\ge \vect{0}_n$,
  $x(0)>\vect{0}_n$, $r(0) \ge \vect{0}_n$ and
  $s(0)+x(0)+r(0)=\vect{1}_n$. Let
  $\big(s(\infty),\vect{0}_n,r(\infty) \big)$ be the asymptotic state
  of system~\eqref{def:network-SIR-vector}. The map
  $\map{H}{\real^n}{\real^n}$ has the following properties:
\begin{enumerate}
\item there exists a unique fixed point $s^*$ of the map $H$ in the
  set $\setdef{s\in \real^n}{\vect{0}_n \le s \le
    \vect{1}_n-r(0)}$. Moreover, $s^*=s(\infty)$ and
  $r(\infty)=\vect{1}_n - s^*$; and
\item any sequence $\{y(k)\}_{k\in \naturals}$ defined by
  $y(k+1)=H(y(k))$ and initial condition $\vect{0}_n \le y(0) \le
  \vect{1}_n-r(0)$ converges to the unique fixed point $s^*$.
\end{enumerate}
\end{theorem}

\begin{proof}  
Since $A$ is a non-negative matrix, and $s(0)\le \vect{1}-r(0)$, one
can easily observe that, if $\vect{0}_n \le p\le q \le
\vect{1}_n-r(0)$, then $\vect{0}_n \le H(\vect{0}_n) \le H(p) \le H(q)
\le H(\vect{1}_n-r(0)) \le \vect{1}_n-r(0)$. According to the Brower
Fixed Point Theorem, the map $H$ has at least one fixed point.

Define the sequence $\{p(k)\}_{k\in \naturals}$ by $p(k+1)=H(p(k))$
and $p(0)=\vect{0}_n$. Since
\begin{equation*}
  \vect{1}_n - r(0) \ge p(1) = H(\vect{0}_n)=e^{ \frac{\beta}{\gamma}\diag\!\big( -A\vect{1}_n + Ar(0) \big) }s(0) \ge p(0),
\end{equation*} 
we have $\vect{1}_n - r(0) \ge p(2)=H(p(1))\ge H(p(0))=p(1)$ and, by
induction, $\vect{1}_n - r(0) \ge p(k+1)\ge p(k)$ for any $k\in
\naturals$. Since $p(k)$ is non-decreasing and upper bounded by
$\vect{1}_n-r(0)$, we conclude that the limit $p^* = \lim_{k\to
  \infty} p(k)$ exists, and $p^*$ is a fixed point of the map $H$.

Similarly, define a sequence $\{q(k)\}_{k\in \naturals}$ by
$q(k+1)=H(q(k))$ and $q(0)=\vect{1}_n-r(0)$. One can check that $q(k)$
is non-increasing and that $q^*=\lim_{k\to \infty} q(k)$ is a fixed
point of map $H$. Moreover, since $p(0)\le q(0)$, we have $p(k)\le
q(k)$ for any $k\in \naturals$ and thereby $p^*\le q^*$.

If $p^*=q^*$, then, for any $\vect{0}_n \le y(0) \le \vect{1}_n-r(0)$,
the sequence $\{y(k)\}_{k\in \naturals}$ defined by $y(k+1)=H(y(k))$
satisfies $p(k)\le y(k) \le q(k)$ for any $k\in \naturals$. Therefore,
$y^*=\lim_{k\to \infty} y(k)$ exists and $y^*=p^*=q^*$, which implies
that the fixed point of map $H$ is unique. According to
equation~\eqref{eq-net-SIR-iteration-asym-s}, $s(\infty)$ is the
unique fixed point. This concludes the proof for statement (i) and
(ii).

Now we eliminate the case $p^*<q^*$ by contradiction. First of all we
prove that $q^* \ll \vect{1}_n-r(0)$. Let $N_i = \{j\,|\, a_{ij}>0\}$
and $\mathcal{I}(k)=\big{\{} i \,\big|\, q_i(\tau)<1-r_i(0) \text{ for
  any }\tau \ge k \big{\}}$. We have $\mathcal{I}(0)=\phi$, and
$\mathcal{I}(1)=\big{\{} i\,\big|\, s_i(0)<1-r_i(0) \big{\}}$, since
$q(k+1)=H(q(k))\le s(0)$ for any $k\in \naturals$. Moreover, since,
for any $i$ such that $N_i\cap \mathcal{I}(k)\neq \phi$,
\begin{equation*}
  q_i(k+1) = H(q(k))_i = e^{ \frac{\beta}{\gamma} \sum_{j=1}^n a_{ij}\big( q_j(k)-1+r_j(0) \big)}s_i(0) < s_i(0) ,
\end{equation*}
we have $\mathcal{I}(k+1)=\big{\{} i \,\big|\, N_i\cap
\mathcal{I}(k)\neq \phi \big{\}} \cup \mathcal{I}(k)$. Because the
graph associated with $A$ is strongly connected, we can argue that
$\mathcal{I}(k)$ contains all the indices when $k$ is large
enough. Therefore, $q^* \ll \vect{1}_n - r(0)$.

Now suppose $p^*<q^*$. Let 
\begin{equation*}
  \alpha = \min_j \frac{1-r_j(0)-p_j^*}{q_j^*-p_j^*},\quad \text{and}\quad w=(1-\alpha)p^* + \alpha q^*. 
\end{equation*}
We have $\alpha>1$, $\vect{0}_n \le w < \vect{1}_n-r(0)$, and
$w_i=1-r_i(0)$ for any $i$ such that $\alpha_i = \big( 1-r_i(0)-p_i^*
\big)/( q_j^*-p_j^* )$. Let $\mu=1/\alpha$. Thereby $q^* = \mu w +
(1-\mu) p^*$, where $0<\mu<1$. This means that $q^*$ is a convex
combination of $p^*$ and $w$. Since $H(s)_i$ is a strictly convex
function of $s$, we obtain that
\begin{align*}
  q_i^*  = H\big( \mu w + (1-\mu) p^* \big)_i < \mu H(w)_i + (1-\mu) p_i^* \le \mu \big( 1-r_i(0) \big) + (1-\mu) p_i^* = q_i^*.
\end{align*}
In the last inequality, we used the fact that $H(w)_i\le 1-r_i(0)$ for
any $0\le w \le \vect{1}_n - r(0)$. The previous inequality yields a
contradiction.
\end{proof}

In the rest of this section, we present some numerical results for the
network SIR model on the undirected unweighted graph illustrated in
Figure~\ref{fig:arbitraryGraph}. The adjacency matrix $A$ is binary.
Unless otherwise stated, the system parameters are $\beta=0.5$ and
$\gamma=0.4$.  As initial condition, we select one node fully infected
(the dark-gray node in Figure~\ref{fig:arbitraryGraph}, say, with
index $1$), 19 fully healthy individuals, and zero recovered fraction
--- corresponding to $x(0)=\vect{e}_1$, $r(0)=\vect{0}_n$, and
$s(0)=\vect{1}_n-x(0)$.  These parameters lead to an initial effective
reproduction number $R(0)=3.57$.

\begin{figure}[h]
  \centering
  \includegraphics[width=.5\linewidth]{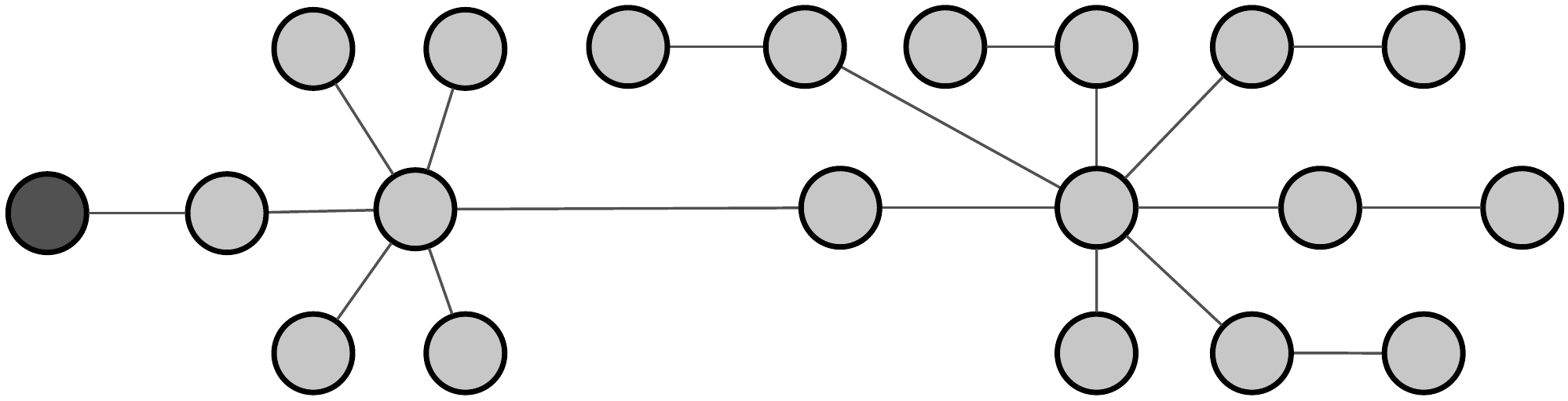}
  \caption{Sample undirected unweighted graph with $20$ nodes}
  \label{fig:arbitraryGraph}
\end{figure}

Figure~\ref{fig:Lambda-maxt} illustrates the time evolution of
$(\beta/\gamma)\lmax(t)$ with varying network parameters. Note that
each evolution starts above the threshold, reaches the threshold value
$1$ in finite time, and converges to a final value below $1$.
\begin{figure}[h]
  \centering
  \includegraphics[width=.45\linewidth]{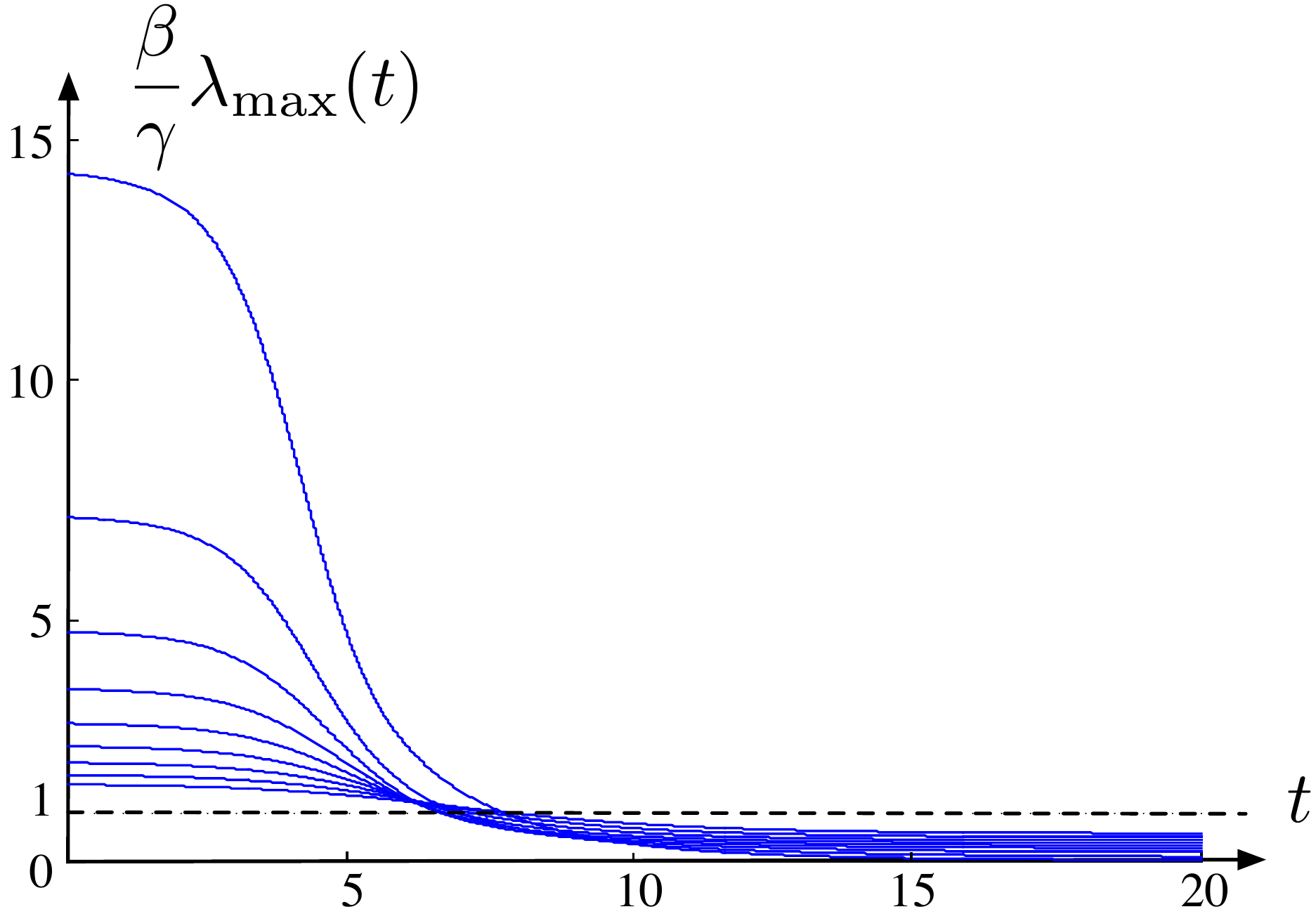}
  \caption{Evolution of the spectral radius of
    $(\beta/\gamma)\diag(s(t))A)$ over the undirected graph in
    Figure~\ref{fig:arbitraryGraph}. The parameter $\gamma$ takes
    value in $.1,.2,\dots,.9$, corresponding respectively to the curves from 
    up to down in the time interval $[0,5]$.}
  \label{fig:Lambda-maxt}
\end{figure}

Figure~\ref{fig:evol-sir-network-arbitrary} illustrates the behavior
of the average susceptible, average infected and average recovered
quantities in populations starting from a small initial infection
fraction and with an effective reproduction number above $1$ at time
$0$. Note that the evolution of the infected fraction of the
population displays a unimodal dependence on time, like in the scalar
model.

\begin{figure}[htb]
  \centering
  \includegraphics[width=.45\linewidth]{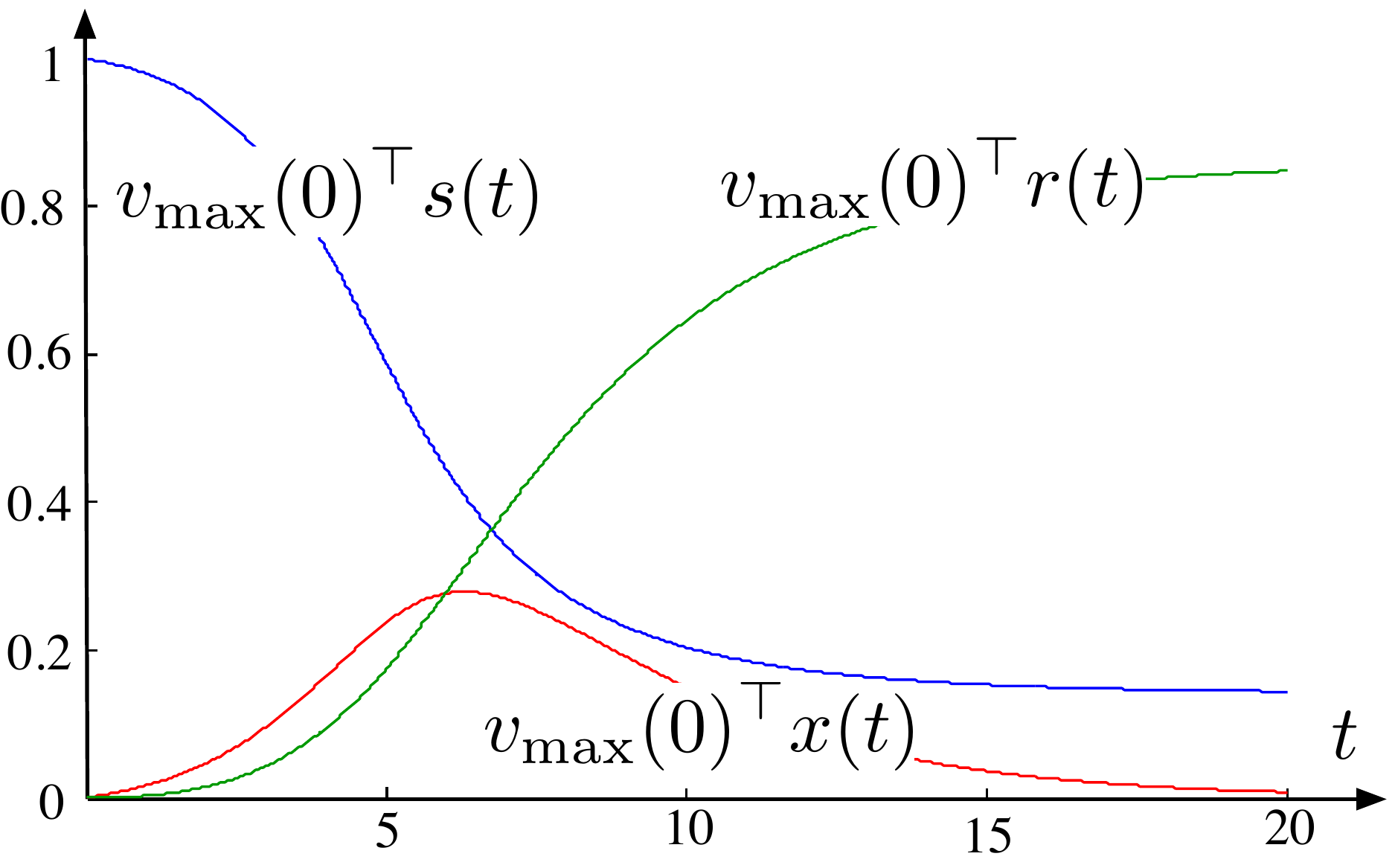}
  \caption{Evolution of the network SIR model from initial condition
    consisting of one node fully infected individual (the dark-gray
    node in Figure~\ref{fig:arbitraryGraph}), 19 fully healthy
    individuals, and zero recovered fraction. The effective
    reproduction number satisfies $R(0)=3.57$.}
  \label{fig:evol-sir-network-arbitrary}
\end{figure}

\section{Conclusion}
This paper provides a comprehensive and consistent treatment of
deterministic nonlinear continuous-time SI, SIS, and SIR propagation
models over contact networks. We investigated the asymptotic
behaviors (vanishing infection, steady-state epidemic, and full
contagion). We studied the transient propagation of an epidemic
starting from small initial fractions of infected nodes. We presented
conditions under which a possible epidemic outbreak occurs or the
infection monotonically vanishes for arbitrary fixing topology
graphs. We introduced a network SI model and analyzed its
behavior. Network SIS model sections includes improved properties over
previously proposed works. New transient behavior, threshold
condition, and system properties for the network SIR model were
proposed. In addition, for the network SIR model, we provide a novel 
iterative algorithm to compute the asymptotic state of the system.
In all cases, we show the results for network models are
appropriate generalizations of those for the respective scalar models.

\bibliographystyle{plain}
\bibliography{alias,Main,FB}

\end{document}